%% file: Driver_GROM.tex
\documentclass[11pt,reqno]{amsproc}
\allowdisplaybreaks

\usepackage{amssymb}
\usepackage{stmaryrd}

\usepackage{enumerate}
\usepackage{fullpage}
\usepackage{textcomp}
\usepackage{lettrine}

\usepackage{graphicx}
\usepackage{subfigure}
\usepackage{float}
\usepackage{morefloats}

\usepackage[font=small,labelfont=bf,
   justification=justified,
   format=plain]{caption} 

\usepackage[debug=false, colorlinks=true, pdfstartview=FitV,
linkcolor=blue, citecolor=blue, urlcolor=blue]{hyperref}
\usepackage[semicolon,square,authoryear,sort]{natbib}

\usepackage[doipre={DOI:~}]{uri}

\usepackage{color}
\usepackage{colortbl}
\usepackage[usenames,dvipsnames,table]{xcolor}
\usepackage[most]{tcolorbox}

\newlength{\drop}
\definecolor{amethyst}{rgb}{0.6, 0.4, 0.8}
\definecolor{burgundy}{rgb}{0.5, 0.0, 0.13}

\usepackage{longtable} 
\usepackage{multirow}
\usepackage{rotating}
\usepackage{bigstrut} 
\usepackage{hhline} 

\usepackage{amsthm}

\newtheoremstyle{remboldstyle}
  {}{}{}{}{\bfseries}{.}{.5em}{{\thmname{#1 }}{\thmnumber{#2}}{\thmnote{ (#3)}}}
\theoremstyle{remboldstyle}

\newtheorem{theorem}{Theorem}[section]

\title{\textbf{Closed-loop geothermal systems: Modeling and predictions}}

\author{\textbf{K.~Adhikari$^{1}$}, \textbf{M.~K.~Mudunuru$^{2}$}, and \textbf{K.~B.~Nakshatrala$^{1}$} \\
  {\small 
  $^{1}$Department of Civil and Environmental Engineering, 
  University of Houston, Houston, TX 77204, USA. \\
  $^{2}$Pacific Northwest National Laboratory, Richland, WA 99354, USA. 
  } \\
  \textbf{Correspondence to:} knakshatrala@uh.edu, 
  +1-713-743-4418
  }


\begin{document}

\begin{titlepage}
  \drop=0.1\textheight
  \centering
  \vspace*{\baselineskip}
  \rule{\textwidth}{1.6pt}\vspace*{-\baselineskip}\vspace*{2pt}
  \rule{\textwidth}{0.4pt}\\[\baselineskip]
       {\Large \textbf{\color{burgundy}
       Closed-loop geothermal systems: Modeling and predictions}}\\[0.3\baselineskip]
       \rule{\textwidth}{0.4pt}\vspace*{-\baselineskip}\vspace{3.2pt}
       \rule{\textwidth}{1.6pt}\\[0.2\baselineskip]
       \scshape
       An e-print of this paper is available on arXiv. \par
       \vspace*{0.2\baselineskip}
       Authored by \\[0.2\baselineskip]

{\Large K.~Adhikari\par}
  {\itshape Graduate Student, Department of Civil \& Environmental Engineering \\
  University of Houston, Houston, Texas 77204.}\\[0.2\baselineskip]

  {\Large M.~K.~Mudunuru\par}
  {\itshape Staff Scientist, 
  Pacific Northwest National Laboratory, Richland, Washington 99354.}\\[0.2\baselineskip]
  
  {\Large K.~B.~Nakshatrala\par}
  {\itshape Department of Civil \& Environmental Engineering \\
  University of Houston, Houston, Texas 77204. \\
  \textbf{phone:} +1-713-743-4418, \textbf{e-mail:} knakshatrala@uh.edu \\
  \textbf{website:} http://www.cive.uh.edu/faculty/nakshatrala}\\[-0.25\baselineskip]
\begin{figure*}[ht]
    \centering
    \includegraphics[scale=0.43]{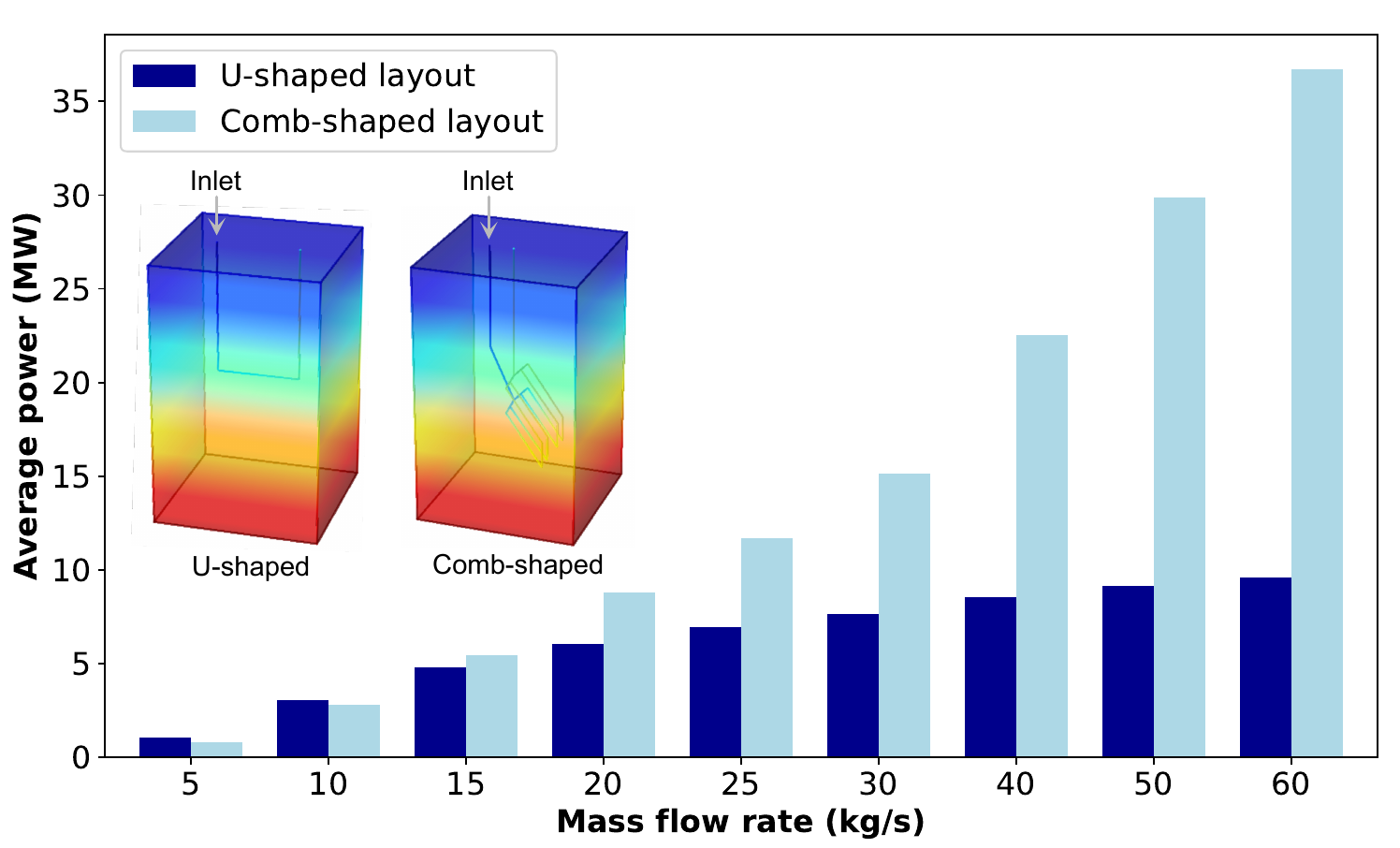}
    \vspace{-0.1in}   
    \captionsetup{labelformat=empty}
    \caption{The figure illustrates the average power production and temperature profile of geothermal systems featuring U-shaped and comb-shaped closed loop layouts. The geothermal gradient is approximately 30 K/km. Fluid enters the system at the inlet, travels through the channel, and exits the outlet at an elevated temperature. At lower mass flow rates, such as 5 kg/s and 10 kg/s, the U-shaped layout provides better thermal power production compared to comb-shaped layout. However, as the mass flow rate increases, the average power generated by the comb-shaped layout increases significantly. Therefore, the right combination of vascular layout and mass flow rate helps in achieving an efficient closed-loop geothermal system.}
\end{figure*}
\vspace{-0.1in}
  \vfill
  {\scshape 2024} \\
  {\small Computational \& Applied Mechanics Laboratory} \par
\end{titlepage}

\begin{abstract}
Geothermal energy is a sustainable baseload source recognized for its ability to provide clean energy on a large scale. \emph{Advanced Geothermal Systems} (AGS)---offer promising prototypes---employ a closed-loop vascular layout that runs deep beneath the Earth’s surface. A working fluid (e.g., water or supercritical carbon dioxide (sCO2)) circulates through the vasculature, entering the subsurface at the inlet and exiting at the outlet with an elevated temperature. For designing and performing cost-benefit analysis before deploying large-scale projects and maintaining efficiency while enabling real-time monitoring during the operational phase, modeling offers cost-effective solutions; often, it is the only available option for performance assessment. A knowledge gap exists due to the lack of a fast predictive modeling framework that considers the vascular intricacies, particularly the jumps in the solution fields across the channel. Noting that the channel diameter is considerably smaller in scale compared to the surrounding geological domain, we develop a reduced-order modeling (ROM) framework for closed-loop geothermal systems. This ROM incorporates the jump conditions and provides a quick and accurate prediction of the temperature field, including the outlet temperature, which directly correlates with the power production capacity and thermal draw-down. We demonstrate the predictive capabilities of the framework by establishing the uniqueness of the solutions and reporting representative numerical solutions. The modeling framework and the predictions reported in this paper benefit the closed-loop geothermal community, enabling them to determine the system’s performance and optimal capacity.
\end{abstract}

\maketitle

\vspace{-0.2in}

\setcounter{figure}{0}   
\setcounter{section}{0} 

\input{Sections/Abbreviation}

\clearpage 
\input{Sections/S1_GROM_Intro}

\input{Sections/S2_GROM_Setup}

\input{Sections/S3_GROM_GE}

\input{Sections/S4_GROM_Properties}

\input{Sections/S5_GROM_NR_Domain_BCs}

\input{Sections/S6_GROM_Efficiency}

\input{Sections/S7_GROM_Surface_MST}

\input{Sections/S8_GROM_Discussion_Closure}

\section*{ACKNOWLEDGMENTS}
KA and KBN acknowledge the support from the UH-Chevron Energy Fellowship, while MKM thanks the U.S. DOE Geothermal Technologies Office (Award \#: DE-EE-3.1.8.1). The opinions expressed in this paper are those of the authors and do not necessarily reflect those of the sponsors.

\section*{DATA AVAILABILITY}
The data supporting this study's findings are available from the corresponding author upon reasonable request.

\bibliographystyle{plainnat}
\bibliography{Master_References}
\end{document}

%% file: Sections/Abbreviation.tex
\section*{A LIST OF ABBREVIATIONS}
\begin{longtable}{ll}\hline
    \hline \multicolumn{2}{c}{\emph{Abbreviations}} \\ \hline
    \textsf{AGS} & advanced geothermal system \\
    \textsf{BC} & boundary condition \\
    \textsf{EGS} & enhanced geothermal system \\
    \textsf{IBVP} & initial boundary value problem \\
    \textsf{ROM} & reduced-order model \\
    \hline
\end{longtable}

\setcounter{table}{0}   

%% file: Sections/S1_GROM_Intro.tex
\section{INTRODUCTION AND MOTIVATION}
\label{Sec:S1_GROM_Intro}

\lettrine[findent=2pt]{\fbox{\textbf{C}}}{onventional energy sources}---originating from fossil fuels such as coal, oil, and natural gas---come with various drawbacks. \emph{First}, the dwindling availability of these natural resources poses a threat to our long-term energy sustainability. \emph{Second}, scientists recognize that energy production reliant on fossil fuels significantly contributes to the emission of greenhouse gases, widely acknowledged as a critical factor in driving climate change \citep{owidco2andgreenhousegasemissions,friedlingstein2022global}. \emph{Third}, extracting, processing, and using these fossil-based energy sources pollute the environment \citep{evans2009assessment}. Therefore, delving into alternative energy technologies and progressively transitioning towards a sustainable, low-carbon energy future is imperative.  

Clean energy enables this transition. Often called renewable or green energy, it encompasses energy sources and technologies with minimal adverse environmental effects compared to traditional fossil fuels. The principal aim of clean energy is to decrease carbon emissions into the atmosphere. Some popular clean energy forms include solar, wind, hydrogen, and geothermal energy \citep{jaiswal2022renewable}. Among these clean energy alternatives, geothermal energy has emerged as a reliable renewable resource, capable of providing base load power during both peak (day) and off-peak (night) times \citep{falcone2018assessment}.

Several factors make geothermal energy truly reliable and sustainable:
\begin{enumerate}
\item The Earth possesses extensive reserves of thermal energy: its core produces a power of approximately 40 $\mathrm{TW}$ \citep{johnston2011emerging,stober2021geothermal}. According to the World Energy Council (WEC), the Earth's crust alone holds a total energy of $5.4 \times 10^{27} \, \mathrm{J}$  \citep{kumari2019sustainable}. A mere $0.1\%$ of this energy could meet the planet's energy demands for 2000 years, considering the current global energy demand of $2.7 \times 10^{21} \, \mathrm{J}$  \citep{bertani2012geothermal}. 
\item Geothermal energy excels in environmental friendliness, leading with a score of 7.23/10 for minimal air pollution among clean energy options \citep{dincer2015review}. 
\item Advancements in drilling technology, primarily from the hydraulic fracturing boom, enable the boring of vertical and horizontal channels spanning miles \citep{bertani2018first}. This flexibility in drilling makes geothermal energy competitive, as the amount of available thermal energy increases with depth. 
\item With the technical know-how gained from various case studies \citep{mohamed2021significance}, the ability to predict drilling fluids' rheology \citep{nasiri2010novel}, and the availability of anti-sagging agents combined with novel drilling techniques \citep{elkatatny2019mitigation}, achieving intricate channel layouts is possible and economical. 
\end{enumerate}
Because of these favorable characteristics and technological developments, as of 2016, twenty-six countries across the globe use geothermal energy to generate electricity \citep{bertani2016geothermal}.

Currently, there are two prominent geothermal prototypes that are scalable: \emph{Enhanced Geothermal Systems} (EGS) and \emph{Advanced Geothermal Systems} (AGS). A typical EGS comprises two parallel channels, an injection well and a production well, extending deep into the Earth's subsurface, spanning orders of kilometers \citep{olasolo2016enhanced,ricks2024role}. The process involves injecting water into the subsurface through the inlet channel. This water then navigates through fractures in hot rock formations, ultimately exiting via the outlet channel. Along this flight, the water absorbs heat and exits the subsurface at an elevated temperature \citep{norbeck2018field}. The resulting hot water or steam can then be harnessed for either household heating or electricity generation. Although the injected water could enhance permeability by reopening pre-existing fractures, it can also risk triggering seismic activity \citep{doeEGS,kaieda2010comparison,breede2013systematic}.

Consequently, AGS, a closed-loop geothermal system, effectively harnesses energy from deeper hot rock layers and low-permeable regions without induced seismicity \citep{van2020new,amaya2020greenfire}. AGS operates by circulating a working fluid from the surface through a network of channels without resorting to reservoir stimulation, thereby reducing seismic risks \citep{malek2022techno,yasukawa2021geothermal}. AGS involves pure conduction with limited contact with the rocks, resulting in lower thermal conductivity. This characteristic leads to a gradual replenishment of thermal energy, ultimately ensuring sustained high thermal performance \citep{wang2021heat,cheng2013studies}. A full-scale demonstration project called Eavor-Lite in Alberta, Canada has been operating since 2019 \citep{toews2021eavor}, and the first commercial closed-loop project called Eavor-loop$^{\mathrm{TM}}$ is recently initiated in Germany \citep{longfield2022eavor}. An exponential growth in various studies and numerical investigation on AGS is seen from 2001 through 2021 \citep{white2024numerical}. Recognizing the advantages and recent advancements \citep{white2024numerical}, our study focuses on modeling closed loop systems.

Given the intricate \emph{subsurface} processes within a geothermal energy system, the utilization of modeling tools is indispensable. Various concomitant strategies encompass economic modeling (including cost-benefit analysis, risk assessment, and regulatory changes), sustainability analysis, reservoir modeling (e.g., focusing on geological aspects like faults), integration with power systems, technology developments (e.g., drilling advancements), and heat extraction modeling (e.g., encompassing well-bore modeling and thermodynamics models) \citep{malek2022techno,beckers2022techno,esmaeilpour2021performance,bottcher2016geoenergy,watanabe2017geoenergy}. This paper explicitly addresses the last aspect of modeling.  

The complexity of a model depends on the specific processes and physical scales one seeks to capture. In the case of heat extraction modeling of a geothermal system, if one wants to resolve the fluid transport within the vasculature, one must use the so-called conjugate heat transfer, which couples the flow dynamics of the fluid (i.e., coolant) and heat transfer in the host solid and fluid \citep{zhang2013conjugate}. However, when the goal is to forecast temperatures in the subsurface and at the outlet, employing a reduced-order model will suffice. The choice is justified, as the diameter of the bore-hole is $0.3 \, \mathrm{m}$ while the depth that the vasculature extends is of order tens of kilometers. Consequently, it is reasonable to assume the vasculature behaves as a curve within the domain. This assumption enables the development of a simplified mathematical model that accurately predicts the outlet temperature and temperature field in the subsurface. Therefore, we present a reduced-order modeling strategy for closed-loop thermal regulation systems that offers fast forecasting of temperature fields, especially the outlet temperature.

To advance modeling efforts and attain a deeper understanding of the behavior of a closed loop geothermal system, answers to several key scientific questions are warranted. The first two pertain to the core aspects of modeling:
\begin{enumerate}[(Q1)]
    \item What should be the appropriate size of the computational domain for a given vascular layout?
\end{enumerate}
\begin{enumerate}[(Q2)]
    \item What are the suitable boundary conditions, and how does this choice affect the prediction of the system's performance?
\end{enumerate}
Sustained performance is essential for the successful operation of an  AGS. Since the outlet temperature relates directly to the thermal power production, the third scientific question reads: 
\begin{enumerate}[(Q3)]
    \item How does the outlet temperature develop over time? Specifically, how does the flow rate influence the initial ascent, the time to reach the peak, and the long-term behavior?   
\end{enumerate}
On the practical front, predicting the top surface temperature will be valuable to stakeholders and policymakers. Ergo, the fourth and final question is:
\begin{enumerate}[(Q4)]
    \item How does the temperature field vary on the top surface change over time? 
\end{enumerate}
This paper adeptly answers these questions, thereby offering a modeling capability as well as conspicuous insights into the behavior of a geothermal system. 

Here is the plan for the rest of this paper. We outline the features of a typical closed-loop geothermal system, provide the problem setup, and discuss the leading vascular layouts (\S\ref{Sec:S2_GROM_Setup}). Next, we present the governing equations of the proposed reduced-order model (ROM) alongside the weak finite element formulation (\S\ref{Sec:S3_GROM_GE}). Following that, we offer a mathematical analysis showing that the solutions under the proposed model are unique (\S\ref{Sec:S4_GROM_Properties}). Then, using the proposed ROM and finite element numerical framework, we present numerical solutions that encompass a convergence study and impact of boundary conditions (\S\ref{Sec:S5_GROM_NR_Domain_BCs}). Subsequently, we analyze the efficiency of geothermal systems from two perspectives: outlet temperature and power production (\S\ref{Sec:S6_GROM_Efficiency}). Further, using the numerical solutions, we study the mean surface temperature of the area near a geothermal system (\S\ref{Sec:S7_GROM_Surface_MST}). Finally, we provide the discussion and conclude the paper by highlighting the main findings and contributions and suggesting possible research extensions (\S\ref{Sec:S8_GROM_Discussion_Closure}). 

\begin{figure}[ht]
    \centering
    \includegraphics[scale=0.11]{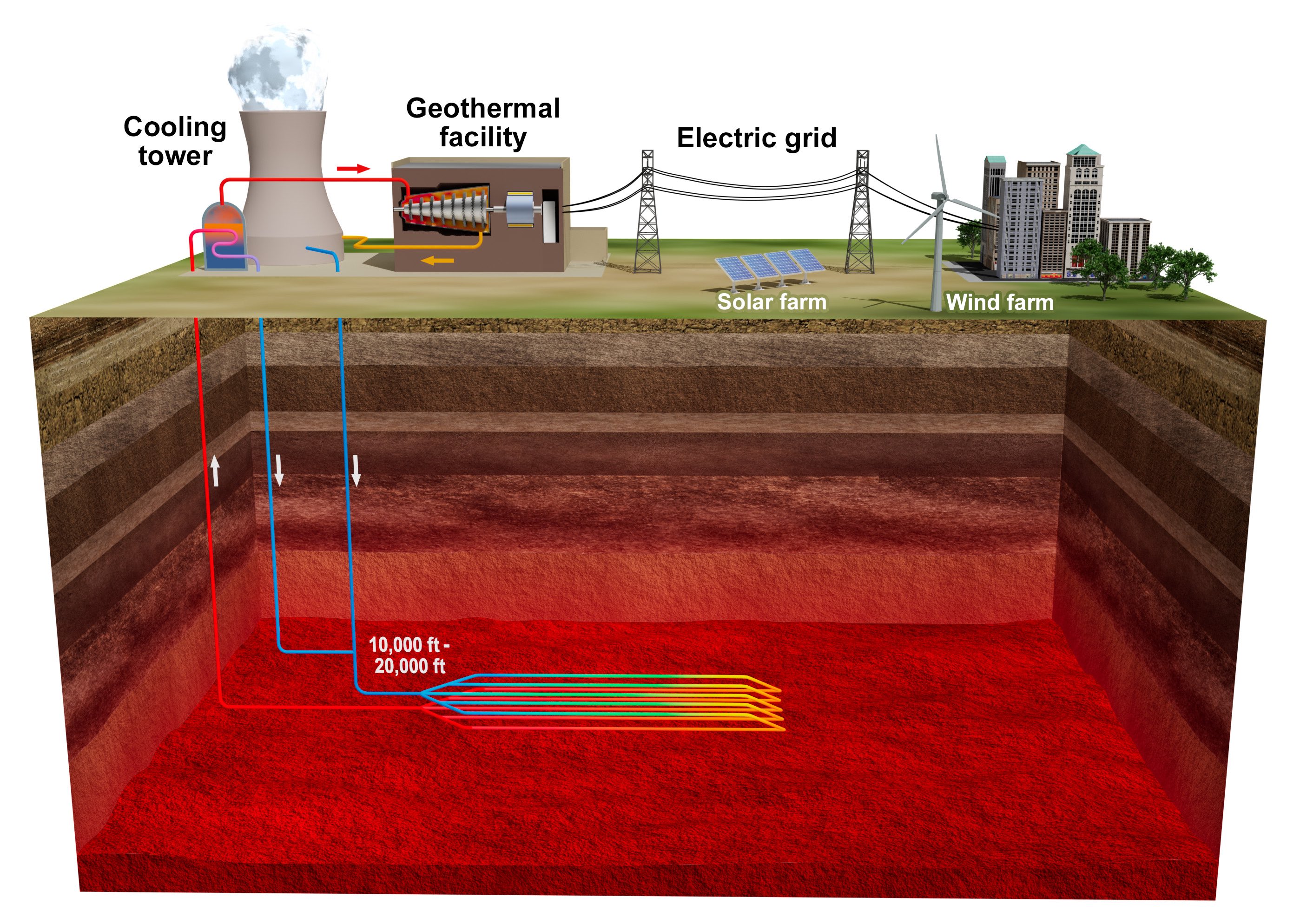}
    \caption{\textsf{A geothermal system:} The figure illustrates a standard closed-loop geothermal system. A coolant, typically water, circulates through an underground network of channels at depths ranging from 10,000 to 20,000 feet, extracting hot water. This heated water undergoes a heat exchange process, transforming into steam. Utilizing this steam, the geothermal facility generates electricity by propelling a turbine. An electric grid facilitated by a network of transmission towers distributes the produced electricity to urban centers. The process releases warm water, which undergoes additional cooling in a dedicated cooling tower, becoming available for other energy consumption (e.g., heating and cooling, agriculture, etc). Additionally, the figure highlights a collaborative approach involving a solar and wind farm to address the region's energy demand. \label{Fig1:GROM_Concept_figure}}
\end{figure}

%% file: Sections/S2_GROM_Setup.tex
\section{VASCULAR LAYOUTS AND PROBLEM SETUP}
\label{Sec:S2_GROM_Setup}

As shown in \textbf{Fig.~ \ref{Fig1:GROM_Concept_figure}}, a geothermal system comprises a gamut of operating machinery (e.g., circulation pump, cooling tower, heat exchanger, electric grid, to name a few) stationed on the surface and a vascular layout beneath the surface, with the inlet and outlet on the horizon. Flowing water (a candidate coolant in geothermal applications) drives the whole heat extraction: A circulation pump pushes water to enter the subsurface at the inlet. The injected water flows through the intermediate channel layout and extracts heat along its flow path from the nearby subsurface. Finally, the heated water exits the subsurface from the outlet.

This paper examines two popular vascular designs, depicted in \textbf{Fig.~\ref{Fig2:GROM_Vascular_layouts}}. The \emph{U-shaped} vascular layout represents a basic channel setup devoid of branching, yet it offers robust performance 
 \citep{sun2018geothermal}. This vasculature ($\Sigma$) comprises two vertical channels linked by a single horizontal channel, forming a U-shaped structure, as shown in Fig.~\ref{Fig2:GROM_Vascular_layouts}A. The inlet is at the start of one of the vertical branches, while the outlet is at the end of the other. The figure also shows the unit tangent vector along the vasculature, denoted by $\widehat{\mathbf{t}}(\mathbf{x})$, and the unit outward normal vectors on either side of a vesicle, represented by  $\widehat{\mathbf{n}}^{\pm}(\mathbf{x})$. 

The \emph{comb-shaped} vascular layout, also popular as Eavor-loop 2.0, is currently considered to be the state-of-the-art geothermal configuration \citep{beckers2022techno}. This layout is increasingly favored because its branching structure covers a larger surface area, significantly enhancing the efficiency of heating the coolant. The concomitant vasculature comprises two vertical legs interconnected by multiple branches, as depicted in Fig.~\ref{Fig2:GROM_Vascular_layouts}B.

\begin{figure}[h]
    \centering
    \includegraphics[scale=0.7]{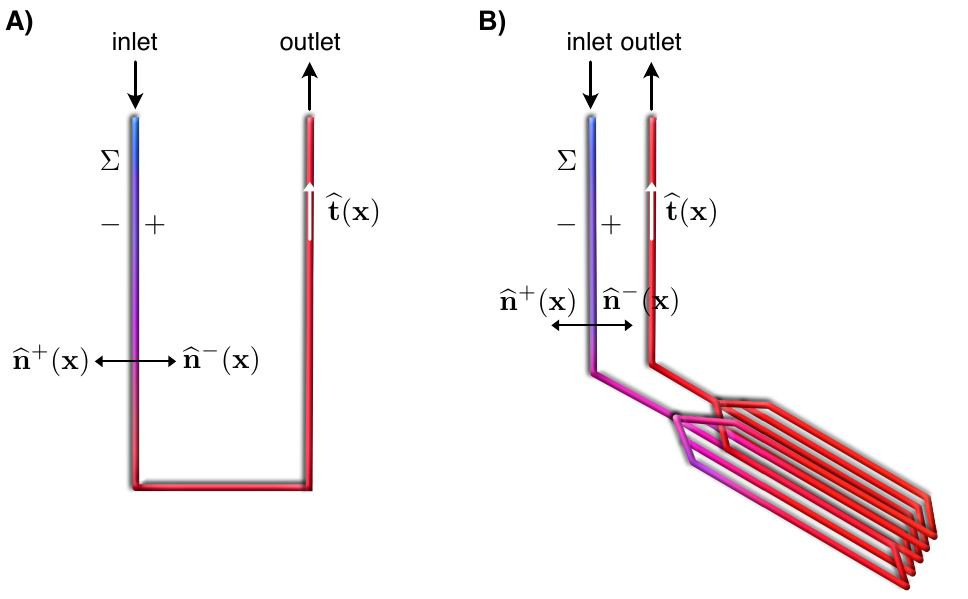}
    \caption{\textsf{Vascular layouts:} This graphic depicts two commonly utilized vascular configurations, both in practical applications and within the context of this paper: \textbf{A)} a U-shaped layout devoid of branching, and \textbf{B)} a comb-shaped arrangement featuring branches. $\Sigma$ denotes the vascular layout, $\widehat{\mathbf{t}}(\mathbf{x})$ the unit tangent vector pointing toward increasing arc-length along the vasculature, and $\widehat{\mathbf{n}}^{\pm}(\mathbf{x})$ the unit outward normal vectors on either side of the vasculature. The fluid (coolant) enters the domain at the inlet and exits at the outlet. As the fluid travels through the length of the vasculature, it extracts heat from the surrounding subsurface, leading to an increase in its thermal energy---either raising the temperature or latent heat.}
    \label{Fig2:GROM_Vascular_layouts}
\end{figure}

In the rest of this paper, we study the performance of these two popular channel designs in extracting thermal energy from the subsurface using modeling.

%% file: Sections/S3_GROM_GE.tex
\section{A GEOTHERMAL MATHEMATICAL MODEL}
\label{Sec:S3_GROM_GE}

We present a mathematical model that predicts the temperature in the subsurface, along the vasculature, and at the outlet. The proposed model is an extension of our previous efforts on thermal regulation in microvascular composites \citep{nakshatrala2023ROM,2023_Nakshatrala_PNAS_Nexus}. Some modeling elements are similar; nevertheless, there are notable differences. 

Our prior studies on microvascular composites considered slender bodies (i.e., thickness is very small compared to other characteristic dimensions); therefore, we pursued a two-dimensional reduced-order model in that application. Such a slenderness assumption is not applicable for geothermal systems (e.g., see the comb-shaped vascular layout given in Figs.~\ref{Fig1:GROM_Concept_figure} and \ref{Fig2:GROM_Vascular_layouts})---geothermal systems span kilometers in all directions. Accordingly, we develop a three-dimensional model for geothermal applications. Nonetheless, noting that the channel cross section is very small compared to the length of the vascular layout and the dimensions of the domain, we still treat the vasculature as a curve, but now as a space curve (i.e., a curve in three dimensions). Also, the boundary conditions are different. In the case of microvascular composites application, the boundaries are adiabatic. In the geothermal application, the lateral and bottom boundaries are prescribed with a temperature field, increasing with depth, while the top surface is free to convect and radiate. 

Consider a domain $\Omega \subset \mathbb{R}^{3}$, bounded by a piece-wise smooth boundary $\partial\Omega$. Mathematically, 
\begin{align}
    \partial \Omega := \overline{\Omega} \setminus \Omega 
\end{align}
in which an overline denotes the set closure. Let $\mathbf{x} \in \overline{\Omega}$ depict a spatial point, and $t \in [0,\mathcal{T}]$ the time, with $\mathcal{T}$ denoting the length of the time interval of interest. $\partial(\cdot)/\partial t$ represents the partial derivative in time, and $\mathrm{div}[\cdot]$ and $\mathrm{grad}[\cdot]$ the spatial divergence and gradient operators, respectively. $\vartheta(\mathbf{x},t)$, a time-dependent scalar field, represents temperature at a spatial point $\mathbf{x}$ and time $t$. Likewise, $\mathbf{q}(\mathbf{x},t)$ is a time-dependent vector field representing the heat flux vector.

The external boundary is divided into three complementary parts: $\Gamma^{\vartheta}$, $\Gamma^{q}$, and $\Gamma^{\mathrm{N}}$. 
That is, 
\begin{align}
    \Gamma^{\vartheta} \cup \Gamma^{q} \cup 
    \Gamma^{\mathrm{N}} = \partial \Omega 
\end{align}
and 
\begin{align}
    \Gamma^{\vartheta} \cap \Gamma^{q} = \emptyset, \;  
    \Gamma^{q} \cap \Gamma^{\mathrm{N}} = \emptyset, \;  
    \mathrm{and} \; 
    \Gamma^{\mathrm{N}} \cap \Gamma^{\vartheta} = \emptyset 
\end{align}
$\Gamma^{\vartheta}$ is the part of the boundary on which temperature (i.e., Dirichlet boundary condition) is prescribed, while $\Gamma^{q}$ is that part of the boundary on which the normal component of the heat flux (i.e., Neumann boundary condition) is prescribed. $\Gamma^{\mathrm{N}}$ is the portion of the boundary that is free to convect and radiate. We use the Stefan-Boltzmann model for radiation, which makes the boundary condition nonlinear (cf. Eq.~\eqref{Eqn:GROM_GE_nonlinear_BC}).

The domain comprises a vasculature $\Sigma$---a connected wellbore within the subsurface with the inlet and outlet on the surface. $s$ denotes the arc-length along the vasculature, starting from the inlet and ending at the outlet. $\widehat{\mathbf{t}}(\mathbf{x})$ represents the unit tangent vector along the vasculature; see \citep[SI]{2023_Nakshatrala_PNAS_Nexus} for a mathematical definition of the tangent vector. We avail the jump and average operators to mathematically describe the continuity of the solution field and the balance of energy across the vasculature. Given a scalar field $\alpha(\mathbf{x},t)$ and a vector field $\mathbf{a}(\mathbf{x},t)$, the jump operator is defined as follows:
\begin{subequations}
\label{Eqn:GROM_Jump_operator}
\begin{align}
    \llbracket \alpha(\mathbf{x},t) \rrbracket 
    &= \alpha^{+}(\mathbf{x},t)
    \, \widehat{\mathbf{n}}^{+}(\mathbf{x})
    + \alpha^{-}(\mathbf{x},t) 
     \, \widehat{\mathbf{n}}^{-}(\mathbf{x}) \\
    \llbracket\mathbf{a}(\mathbf{x},t) \rrbracket 
    &= \mathbf{a}^{+}(\mathbf{x},t)
    \bullet \widehat{\mathbf{n}}^{+}(\mathbf{x})
    + \mathbf{a}^{-}(\mathbf{x},t) 
     \bullet \widehat{\mathbf{n}}^{-}(\mathbf{x})
\end{align}
\end{subequations}
where $\pm$ denotes the limits of the fields from either side of $\Sigma$. For these fields, the average operator is defined as follows:
\begin{subequations}
\label{Eqn:GROM_Average_operator}
\begin{align}
    \{\!\!\{ \alpha(\mathbf{x},t) \rrbracket 
    &= \frac{1}{2} \Big(\alpha^{+}(\mathbf{x},t)
    + \alpha^{-}(\mathbf{x},t)\Big) \\
    \{\!\!\{\mathbf{a}(\mathbf{x},t) \}\!\!\}
    &= \frac{1}{2} \Big( \mathbf{a}^{+}(\mathbf{x},t)
    + \mathbf{a}^{-}(\mathbf{x},t) \Big) 
\end{align}
\end{subequations}
The following property satisfied by the jump and average operators will be handy in establishing the uniqueness of solutions:
\begin{align}
    \label{jump_operator_product}
    \big[\!\!\big[\alpha (\mathbf{x},t) \,\mathbf{a}(\mathbf{x},t) 
    \big]\!\!\big] =
    \{\!\!\{\alpha (\mathbf{x},t)\}\!\!\} 
    \big[\!\!\big[\mathbf{a}(\mathbf{x},t) 
    \big]\!\!\big]
    + \big[\!\!\big[\alpha (\mathbf{x},t) 
    \big]\!\!\big] \bullet \{\!\!\{\mathbf{a}(\mathbf{x},t)\}\!\!\} 
\end{align}
For more details on jump and average operators, refer to \citep{nakshatrala2023thermal}.

The governing equations take the following form:
\begin{subequations}
\begin{alignat}{2}
    \label{Eqn:GROM_GE_BoE}
    &\rho_{s} \, c_{s} \frac{\partial \vartheta(\mathbf{x},t)}{\partial t} + \mathrm{div}\big[\mathbf{q}(\mathbf{x},t)\big] = 0 
    && \quad \mathrm{in} \; \Omega \\
    \label{Eqn:GROM_GE_fourier}
    &\mathbf{q}(\mathbf{x},t) 
    = - \mathbf{K}(\mathbf{x}) \, \mathrm{grad}\big[\vartheta(\mathbf{x},t)\big] 
    && \quad \mathrm{in} \; \Omega \\
    \label{Eqn:GROM_GE_temp_BC}
    &\vartheta(\mathbf{x},t) = \vartheta^{\mathrm{p}}(\mathbf{x},t) 
    && \quad \mathrm{on} \; \Gamma^{\vartheta} \\
    \label{Eqn:GROM_GE_heat_flux_BC}
    &\mathbf{q}(\mathbf{x},t) \bullet \widehat{\mathbf{n}}(\mathbf{x}) 
    = q^{\mathrm{p}}(\mathbf{x},t) 
    && \quad \mathrm{on} \; \Gamma^{q}  \\
    \label{Eqn:GROM_GE_nonlinear_BC}
    &\mathbf{q}(\mathbf{x},t) \bullet \widehat{\mathbf{n}}(\mathbf{x}) 
    = h_T \, \big(\vartheta(\mathbf{x},t) - \vartheta_{\mathrm{amb}} \big) + \varepsilon \, \sigma \,
    \big(\vartheta^{4}(\mathbf{x},t) - \vartheta^{4}_{\mathrm{amb}} \big)
    && \quad \mathrm{on} \; \Gamma^{\mathrm{N}}  \\
    \label{Eqn:GROM_GE_temp_jump_condition}
    &\llbracket\vartheta(\mathbf{x},t)\rrbracket = \mathbf{0}
    && \quad \mathrm{on} \; \Sigma  \\
    \label{Eqn:GROM_GE_q_jump_condition}
    &\llbracket\mathbf{q}(\mathbf{x},t)\rrbracket = \dot{m} \, c_f \,  \mathrm{grad}[\vartheta(\mathbf{x},t)] \bullet \widehat{\mathbf{t}}(\mathbf{x}) 
    && \quad \mathrm{on} \; \Sigma  \\
    \label{Eqn:GROM_GE_IC}
    &\vartheta(\mathbf{x},0) = \vartheta_{\mathrm{initial}}(\mathbf{x})
    &&\quad \mathrm{in} \; \Omega \\
    \label{Eqn:GROM_GE_Inlet}
    &\vartheta(\mathbf{x},t) \big\vert_{s = 0} = \vartheta_{\mathrm{inlet}}
    &&\quad \mbox{at inlet} 
\end{alignat}
\end{subequations}
where $\rho_{s}$ and $c_{s}$ are the density and specific heat capacity of the host solid, respectively. $\mathbf{K}(\mathbf{x})$ is the conductivity of the host solid. $\widehat{\mathbf{n}}(\mathbf{x})$ is the unit outward normal vector. $h_{T}$ denotes the heat transfer coefficient, $\varepsilon$ the emissivity, and $\sigma \approx 5.67 \times 10^{-8} \; \mathrm{W}  \mathrm{m}^{-2} \mathrm{K}^{-4}$ the Stefan-Boltzmann constant. $\vartheta_{\mathrm{amb}}$ is the ambient temperature. 

The fluid's mass flow rate is denoted by 
\begin{align}
\dot{m} = \rho_f \, Q
\end{align}
where $Q$ is the volumetric flow rate and $\rho_{f}$ the fluid's density. The heat capacity rate is defined as 
\begin{align}
\chi = \dot{m} \, c_f 
\end{align}
in which $c_f$ represents the specific heat capacity of the fluid. 

Equation \eqref{Eqn:GROM_GE_BoE} represents the energy balance, while Eq.~\eqref{Eqn:GROM_GE_fourier} is the Fourier law of heat conduction. Equations \eqref{Eqn:GROM_GE_temp_BC} and \eqref{Eqn:GROM_GE_heat_flux_BC} depict the Dirichlet and Neumann boundary conditions, respectively, and Eq. \eqref{Eqn:GROM_GE_nonlinear_BC} is the non-linear boundary condition accounting for the convection (based on Newton's law of cooling) and radiation (based on the Stefan-Boltzmann law). Equations \eqref{Eqn:GROM_GE_temp_jump_condition} and \eqref{Eqn:GROM_GE_q_jump_condition} are the jump conditions, describing the temperature continuity and the balance of heat flux across the vasculature, respectively. Equation \eqref{Eqn:GROM_GE_IC} denotes the initial condition with $\vartheta_{\mathrm{initial}}$ representing the prescribed initial temperature. Equation \eqref{Eqn:GROM_GE_Inlet} is the prescription of the temperature at the inlet, which is equal to the ambient temperature (i.e., $\vartheta_{\mathrm{inlet}} = \vartheta_{\mathrm{amb}}$). 

The resulting governing equations form a nonlinear second-order parabolic initial boundary value problem. The nonlinearity arises due to the the radiation term present in the boundary condition \eqref{Eqn:GROM_GE_nonlinear_BC}

\subsection{Coefficient of performance}
To quantify the performance of a geothermal system, we consider the following coefficient of performance: 
\begin{align}
\label{Eqn:GROM_coeff_of_performance}
\zeta(t) = 1 - \frac{\vartheta_{\mathrm{inlet}}}{\vartheta_{\mathrm{outlet}}(t)}
\end{align}
where the $\vartheta_{\mathrm{inlet}}$ is the prescribed temperature at the inlet, and ${\vartheta_{\mathrm{outlet}}(t)}$ is the outlet temperature---a time varying quantity. Thus, the higher the outlet temperature, the higher the performance. 

\subsection{Weak formulation} For obtaining numerical solutions, we have employed the finite element method (FEM) based on the single-field Galerkin formulation. We denote the weighting (or test) function by $\xi(\mathbf{x})$ and define the associated function spaces as follows:
\begin{align}
    \label{Eqn:GROM_trial_function_space}
    \mathcal{U}_{t} &:= \left\{\vartheta(\mathbf{x},t) \in H^{1}(\Omega) \; \big\vert \; \vartheta(\mathbf{x},t) = \vartheta^\mathrm{p}(\mathbf{x},t) \; \mathrm{on} \; \Gamma^\vartheta 
    \; \mathrm{and} \; 
    \vartheta(\mathbf{x},t) \big\vert_{s=0} = \vartheta_{\mathrm{inlet}} \right\} \\
    \label{Eqn:TProp_test_function_space}
    \mathcal{W} &:= \left\{\xi(\mathbf{x}) \in H^{1}(\Omega) \; \big\vert \; \xi(\mathbf{x}) = 0 \; \mathrm{on}\; \Gamma^\vartheta 
    \; \mathrm{and} \; 
    \xi(\mathbf{x}) \big\vert_{s=0} = 0
    \right\}
\end{align}
where $H^{1}(\Omega)$ is a standard Hilbert space \citep{brezzi2012mixed}. The Galerkin weak formulation reads: Find $\vartheta(\mathbf{x},t) \in \mathcal{U}_{t}$ such that we have 
\begin{align}
    \label{Eqn:GROM_Galerkin_formulation}
    \int_{\Omega}\rho_{s} \, c_{s}\, \xi(\mathbf{x})\,\frac{\partial \vartheta(\mathbf{x},t)}{\partial t}\, \mathrm{d}\Omega
    &+ \int_{\Omega} \mathrm{grad}[\xi(\mathbf{x})] \bullet \mathbf{K}(\mathbf{x}) \, \mathrm{grad}[\vartheta(\mathbf{x},t)] \,  \mathrm{d} \Omega \nonumber \\
    &+\int_{\Sigma} \chi\, \xi(\mathbf{x})\;\mathrm{grad}[\vartheta(\mathbf{x},t)]\bullet \widehat{\mathbf{t}}(\mathbf{x})\, \mathrm{d}\Gamma 
    \nonumber \\
    &+\int_{\Gamma^{\mathrm{N}}} \xi(\mathbf{x})\; 
    \Big(h_T \, \big(\vartheta(\mathbf{x},t) - \vartheta_{\mathrm{amb}} \big) + \varepsilon \, \sigma \,
    \big(\vartheta^{4}(\mathbf{x},t) - \vartheta^{4}_{\mathrm{amb}} \big)\Big) \, \mathrm{d}\Gamma \nonumber \\ 
    &\qquad = - \int_{\Gamma^q} \xi(\mathbf{x})\; q^\mathrm{p}(\mathbf{x},t)\, \mathrm{d}\Gamma
    \qquad \forall \xi(\mathbf{x})\in \mathcal{W}
\end{align}

We have implemented the above weak formulation using the ``Weak Form PDE'' capability available in COMSOL \citep{COMSOL}. This feature in COMSOL approximates the user-provided weak form using the FEM for spatial discretization and the backward difference formula (BDF) for temporal discretization. We have used second-order Lagrange elements and BDF with a minimum order of one to a maximum of five. 

%% file: Sections/S4_GROM_Properties.tex
\section{UNIQUENESS OF SOLUTIONS}
\label{Sec:S4_GROM_Properties}

Before we report numerical results, it is imperative to assess whether the model actually renders unique solutions. What we show in this section is that the solutions are unique when radiation is neglected, while only the ``non-negative'' solutions are unique when the radiation is present. We first present the relevant notation and preparatory remarks that make the derivation succinct. 

On the contrary, if the solutions are not unique, then there exist two time-dependent fields $\vartheta_{1}(\mathbf{x},t)$ and $\vartheta_{2}(\mathbf{x},t)$ each satisfying the governing equations \eqref{Eqn:GROM_GE_BoE}--\eqref{Eqn:GROM_GE_IC}. Consequently, their difference
\begin{align}
    \label{Eqn:GROM_w_definition}
    w(\mathbf{x},t) := 
    \vartheta_{1}(\mathbf{x},t) 
    - \vartheta_{2}(\mathbf{x},t) 
\end{align}
satisfies the following initial boundary value problem:
\begin{subequations}
\begin{alignat}{2}
    \label{Eqn:Temp_GE_BoE_w}
    &\rho_{s} \, c_{s} \frac{\partial w(\mathbf{x},t)}{\partial t} 
    - \mathrm{div}\big[\mathbf{K}(\mathbf{x}) \, \mathrm{grad}\big[w(\mathbf{x},t)\big]\big] = 0 
    && \quad \mathrm{in} \; \Omega \\
    \label{Eqn:Temp_GE_temperature_BC_w}
    &w(\mathbf{x},t) = 0
    && \quad \mathrm{on} \; \Gamma^{\vartheta} \\
    \label{Eqn:Temp_GE_heat_flux_BC_w}
    -&\widehat{\mathbf{n}}(\mathbf{x}) \bullet 
    \mathbf{K}(\mathbf{x}) \, \mathrm{grad}\big[w(\mathbf{x},t)\big] 
    = 0
    && \quad \mathrm{on} \; \Gamma^{q}  \\
    \label{Eqn:Temp_GE_non-linear_BC_w}
    -&\widehat{\mathbf{n}}(\mathbf{x}) \bullet 
    \mathbf{K}(\mathbf{x}) \, \mathrm{grad}\big[w(\mathbf{x},t)\big] 
    = h_T \, w(\mathbf{x},t) + \varepsilon \, \sigma \,
    \big(\vartheta^{4}_1(\mathbf{x},t) - \vartheta^{4}_2(\mathbf{x},t) \big)
    && \quad \mathrm{on} \; \Gamma^{\mathrm{N}}  \\
    \label{Eqn:Temp_GE_temp_jump_condition_w}
    &\llbracket w(\mathbf{x},t)\rrbracket = \mathbf{0} 
    && \quad \mathrm{on} \; \Sigma  \\
    \label{Eqn:Temp_GE_flux_jump_condition_w}
    -&\big\llbracket\mathbf{K}(\mathbf{x}) \, \mathrm{grad}\big[w(\mathbf{x},t)\big]\big\rrbracket = \dot{m} \, c_f \,  \mathrm{grad}\big[w(\mathbf{x},t)\big] \bullet \widehat{\mathbf{t}}(\mathbf{x}) 
    && \quad \mathrm{on} \; \Sigma  \\
    \label{Eqn:GeoThermal_GE_IC_w}
    &w(\mathbf{x},0) = 0
    &&\quad \mathrm{in} \; \Omega \\
    \label{Eqn:GeoThermal_GE_Inlet_w}
    &w(\mathbf{x},t)\big\vert_{s=0} = 0
    &&\quad \mbox{at inlet} 
\end{alignat}
\end{subequations}

Thus, to establish uniqueness, it suffices to show that $w(\mathbf{x},t) = 0$ is the \emph{only} solution to the above initial boundary value problem \eqref{Eqn:Temp_GE_BoE_w}--\eqref{Eqn:GeoThermal_GE_Inlet_w}. The next two theorems precisely show that the difference $w(\mathbf{x},t)$ vanishes by considering two separate cases: with and without radiation. 

\begin{theorem}[Uniqueness neglecting radiation]
    \label{Thm:GROM_Uniqueness_no_rad}
    Without radiation (i.e., $\varepsilon = 0$), 
    solutions under the mathematical model are 
    unique.
\end{theorem}
\begin{proof} 
    Multiplying Eq.~\eqref{Eqn:Temp_GE_BoE_w} with 
    $w(\mathbf{x},t)$ and integrating over the 
    domain, we get: 
    \begin{alignat}{2}
        \label{Eqn:GROM_Uniqueness_norad_step_1}
        \int_{\Omega}
        w(\mathbf{x},t) \, 
        \rho_{s} \, c_{s} \frac{\partial w(\mathbf{x},t)}{\partial t} 
        \, \mathrm{d} \Omega 
        - \int_{\Omega}
        w(\mathbf{x},t) \, \mathrm{div}\big[\mathbf{K}(\mathbf{x}) \, \mathrm{grad}[w(\mathbf{x},t)]\big] \, 
        \mathrm{d} \Omega 
        = 0 
    \end{alignat}
    Noting that $\rho_s$ and $c_s$ are independent 
    of the temperature and time, we write: 
    \begin{align}
        \label{Eqn:GROM_Uniqueness_norad_step_2}
        \frac{\mathrm{d}}{\mathrm{d}t}\int_{\Omega}\frac{1}{2} 
        \, \rho_{s} \, c_{s} \, w^{2}(\mathbf{x},t) 
        \, \mathrm{d} \Omega 
        - \int_{\Omega}
        w(\mathbf{x},t) \, \mathrm{div}\big[\mathbf{K}(\mathbf{x}) 
        \, \mathrm{grad}[w(\mathbf{x},t)]\big] \, 
        \mathrm{d} \Omega 
        = 0 
    \end{align}
Applying the Green's identity to the second 
term in the above equation, we have: 
\begin{align}
    \label{Eqn:GROM_Uniqueness_norad_step_3}
    \frac{\mathrm{d}}{\mathrm{d}t}\int_{\Omega}\frac{1}{2} \, \rho_{s} \, c_{s} \, w^{2}(\mathbf{x},t) 
    \, \mathrm{d} \Omega 
    &+ \int_{\Omega}
    \mathrm{grad}\big[w(\mathbf{x},t)\big] \bullet \mathbf{K}(\mathbf{x}) \, \mathrm{grad}[w(\mathbf{x},t)] \, 
    \mathrm{d} \Omega \nonumber \\ 
    &- \int_{\partial \Omega}
    w(\mathbf{x},t) \, \widehat{\mathbf{n}}(\mathbf{x}) \bullet \mathbf{K}(\mathbf{x}) \, \mathrm{grad}[w(\mathbf{x},t)] \, 
    \mathrm{d} \Gamma \nonumber \\ 
    &- \int_{\Sigma}
    \big[\!\!\big[
    w(\mathbf{x},t) \, 
    \mathbf{K}(\mathbf{x}) \, \mathrm{grad}[w(\mathbf{x},t)] 
    \big]\!\!\big]
    \, \mathrm{d} \Gamma 
    = 0 
\end{align}
We note that the boundary $\partial\Omega$ is decomposed into $\Gamma^\vartheta$, $\Gamma^q$ and $\Gamma^\mathrm{N}$. Invoking Eqs.~\eqref{Eqn:Temp_GE_temperature_BC_w}--\eqref{Eqn:Temp_GE_non-linear_BC_w} to simplify the third term in the above equation, we arrive at the following: 
\begin{align}
    \label{Eqn:GROM_Uniqueness_norad_step_4}
    &\frac{\mathrm{d}}{\mathrm{d}t}\int_{\Omega} \frac{1}{2} \,\rho_{s} \, c_{s} \, w^{2}(\mathbf{x},t) 
    \, \mathrm{d} \Omega 
    + \int_{\Omega}
    \mathrm{grad}\big[w(\mathbf{x},t)\big] \bullet \mathbf{K}(\mathbf{x}) \, \mathrm{grad}[w(\mathbf{x},t)] \, 
    \mathrm{d} \Omega \nonumber \\ 
    &\qquad + \int_{\Gamma^\mathrm{N}}
    h_T \,w^{2}(\mathbf{x},t) \,\mathrm{d} \Gamma 
    + \int_{\Gamma^\mathrm{N}} w(\mathbf{x},t) \, \varepsilon \, \sigma \,\big(\vartheta^{4}_1(\mathbf{x},t) - \vartheta^{4}_2(\mathbf{x},t) \big) \, 
    \mathrm{d} \Gamma \nonumber\\
    &\qquad- \int_{\Sigma}
    \big[\!\!\big[
    w(\mathbf{x},t) \, 
    \mathbf{K}(\mathbf{x}) \, \mathrm{grad}[w(\mathbf{x},t)] 
    \big]\!\!\big]\,
    \mathrm{d} \Gamma
    = 0 
\end{align}

Since $\mathbf{K}(\mathbf{x})$ is positive definite, the second integral in the above equation is non-negative. The third integral is also non-negative as $h_T \geq 0$ and $w^{2}(\mathbf{x},t) \geq 0$. We, therefore, end up with the following inequality:
\begin{align}
    \label{Eqn:GROM_Uniqueness_norad_step_5}
    \frac{\mathrm{d}}{\mathrm{d}t}\int_{\Omega} \frac{1}{2} \,\rho_{s} \, c_{s} \, w^{2}(\mathbf{x},t) 
    \, \mathrm{d} \Omega 
    &+ \int_{\Gamma^\mathrm{N}} w(\mathbf{x},t) \, \varepsilon \, \sigma \,\big(\vartheta^{4}_1(\mathbf{x},t) - \vartheta^{4}_2(\mathbf{x},t) \big) \, 
    \mathrm{d} \Gamma \nonumber\\
    &\qquad- \int_{\Sigma}
    \big[\!\!\big[
    w(\mathbf{x},t) \, 
    \mathbf{K}(\mathbf{x}) \, \mathrm{grad}[w(\mathbf{x},t)] 
    \big]\!\!\big]\,
    \mathrm{d} \Gamma
    \leq 0 
\end{align}
Applying the property given by Eq.~\eqref{jump_operator_product} to the third term, we get:
\begin{align}
    \label{Eqn:GROM_Uniqueness_norad_step_6}
    \frac{\mathrm{d}}{\mathrm{d}t}\int_{\Omega}\frac{1}{2} \, \rho_{s} \, c_{s} \, w^{2}(\mathbf{x},t) 
    \, \mathrm{d} \Omega 
    &+ \int_{\Gamma^\mathrm{N}} w(\mathbf{x},t) \, \varepsilon \, \sigma \,\big(\vartheta^{4}_1(\mathbf{x},t) - \vartheta^{4}_2(\mathbf{x},t) \big) \, 
    \mathrm{d} \Gamma \nonumber\\
    &- \int_{\Sigma}
    \{\!\!\{w(\mathbf{x},t)\}
    \big[\!\!\big[\mathbf{K}(\mathbf{x}) \, \mathrm{grad}[w(\mathbf{x},t)] 
    \big]\!\!\big] 
    \,
    \mathrm{d} \Gamma \nonumber \\
    &- \int_{\Sigma}\big[\!\!\big[w(\mathbf{x},t) 
    \big]\!\!\big] \bullet 
    \{\!\!\{\mathbf{K}(\mathbf{x}) \, \mathrm{grad}[w(\mathbf{x},t)] \}\!\!\} \,
    \mathrm{d} \Gamma
    \leq 0 
\end{align}
The jump condition \eqref{Eqn:Temp_GE_temp_jump_condition_w} implies that the fourth term vanishes. Further, using the other jump condition \eqref{Eqn:Temp_GE_flux_jump_condition_w} in the third term, we write:
\begin{align}
    \label{Eqn:GROM_Uniqueness_norad_step_7}
    \frac{\mathrm{d}}{\mathrm{d}t}\int_{\Omega}\frac{1}{2} \, \rho_{s} \, c_{s} \, w^{2}(\mathbf{x},t) 
    \, \mathrm{d} \Omega 
    &+ \int_{\Gamma^\mathrm{N}} w(\mathbf{x},t) \, \varepsilon \, \sigma \,\big(\vartheta^{4}_1(\mathbf{x},t) - \vartheta^{4}_2(\mathbf{x},t) \big) \, 
    \mathrm{d} \Gamma \nonumber\\
    &+ \int_{\Sigma}\dot{m} \, c_f \,  \{\!\!\{ w(\mathbf{x},t)\}\!\!\} \, \mathrm{grad}\big[w(\mathbf{x},t)\big] \bullet \widehat{\mathbf{t}}(\mathbf{x}) \,
    \mathrm{d} \Gamma 
    \leq 0 
\end{align}
Noting that $\dot{m}$ and $c_{f}$ are constants, integrating the third term along the vasculature, and using Eq.~\eqref{Eqn:GeoThermal_GE_Inlet_w}, we establish the following inequality:
\begin{align}
    \label{Eqn:GROM_Uniqueness_norad_step_8}
    \frac{\mathrm{d}}{\mathrm{d}t}\int_{\Omega}\frac{1}{2} \, \rho_{s} \, c_{s} \, w^{2}(\mathbf{x},t) 
    \, \mathrm{d} \Omega 
    + \int_{\Gamma^\mathrm{N}} w(\mathbf{x},t) \, \varepsilon \, \sigma \,\big(\vartheta^{4}_1(\mathbf{x},t) - \vartheta^{4}_2(\mathbf{x},t) \big) \, 
    \mathrm{d} \Gamma 
    +\frac{\dot{m} \, c_f}{2} \,   w^2(\mathbf{x},t) \Big\vert_{\mathrm{outlet}}
    \leq 0 
\end{align}
Since $\dot{m}$ and $c_f$ are non-negative, the third term is non-negative, giving rise to the following inequality:
\begin{align}
    \label{Eqn:GROM_Uniqueness_norad_step_9}
    \frac{\mathrm{d}}{\mathrm{d}t}\int_{\Omega}\frac{1}{2} \, \rho_{s} \, c_{s} \, w^{2}(\mathbf{x},t) 
    \, \mathrm{d} \Omega 
    + \int_{\Gamma^\mathrm{N}} w(\mathbf{x},t) \, \varepsilon \, \sigma \,\big(\vartheta^{4}_1(\mathbf{x},t) - \vartheta^{4}_2(\mathbf{x},t) \big) \, 
    \mathrm{d} \Gamma 
    \leq 0 
\end{align}

We now specialize in the vanishing radiation case. Substituting $\varepsilon = 0$ in the above inequality, we get:  
\begin{align}
    \label{Eqn:GROM_Uniqueness_norad_step_10}
    \frac{\mathrm{d}}{\mathrm{d}t}\int_{\Omega}\frac{1}{2} \, \rho_{s} \, c_{s} \, w^{2}(\mathbf{x},t) 
    \, \mathrm{d} \Omega 
    \leq 0 
\end{align}
For convenience, let us define: 
\begin{align}
    \label{Eqn:GROM_Uniqueness_norad_step_11}
    \mathcal{E}(t) := \int_{\Omega}\frac{1}{2} \,                
    \rho_{s} \, c_{s} \, w^{2}(\mathbf{x},t) 
    \, \mathrm{d} \Omega
\end{align}
Then inequality \eqref{Eqn:GROM_Uniqueness_norad_step_10} can be compactly written as follows: 
\begin{align}
    \label{Eqn:GROM_Uniqueness_norad_step_12}
    \frac{\mathrm{d}\mathcal{E}(t)}{\mathrm{d}t} \leq 0
\end{align}
From the property of integration \citep{bartle2000introduction}, 
the above inequality implies that 
\begin{align}
    \label{Eqn:GROM_Uniqueness_norad_step_13}
    \mathcal{E}(t) \leq \mathcal{E}(t_0) \quad \forall t \leq t_0 
\end{align}
Noting that $\mathcal{E}(t_0) = 0$ as $w(\mathbf{x},t=0) = w_{\mathrm{initial}} = 0$, we arrive at the following:
\begin{align}
    \label{Eqn:GROM_Uniqueness_norad_step_14}
    \mathcal{E}(t) \leq 0
\end{align}
But, $\mathcal{E}(t)$ is a non-negative quantity, as the integrand defining this quantity (cf. Eq.~\eqref{Eqn:GROM_Uniqueness_norad_step_11}) is non-negative---$\rho_s$ and $c_s$ are positive, and $w^{2}(\mathbf{x},t)$ is a non-negative field. This discourse implies 
that 
\begin{align}
    \label{Eqn:GROM_Uniqueness_norad_step_15}
    \mathcal{E}(t) = 0
\end{align}
which further implies that 
\begin{align}
    \label{Eqn:GROM_Uniqueness_norad_step_16}
    w(\mathbf{x},t) = 0 \quad \forall \mathbf{x}, \forall t
\end{align}
Noting that $w(\mathbf{x},t) := \vartheta_{1}(\mathbf{x},t) - \vartheta_{2}(\mathbf{x},t)$, we have established the uniqueness of the solutions in the absence of radiation.
\end{proof}

\begin{theorem}[Uniqueness considering radiation]
\label{proof:uniqueness_proof_rad}
With radiation (i.e., $\varepsilon \neq 0$), the \emph{non-negative} solutions under the mathematical model are unique.
\end{theorem}
\begin{proof}
    Most of the steps remain the same as in Theorem \ref{Thm:GROM_Uniqueness_no_rad}. We take up from Eq.~\eqref{Eqn:GROM_Uniqueness_norad_step_9}: 
    \begin{align}
        \label{Eqn:GROM_Uniqueness_rad_step_1}
        \frac{\mathrm{d}}{\mathrm{d}t}\int_{\Omega}\frac{1}{2} \, \rho_{s} \, c_{s} \, w^{2}(\mathbf{x},t) 
       \, \mathrm{d} \Omega 
        + \int_{\Gamma^\mathrm{N}} w(\mathbf{x},t) \, \varepsilon \, \sigma \, \big(\vartheta^{4}_1(\mathbf{x},t) - \vartheta^{4}_2(\mathbf{x},t) \big) \,
        \mathrm{d} \Gamma
       \leq 0 
\end{align}
Expanding the integrand in the second integral and noting that $w(\mathbf{x},t) := \vartheta_{1}(\mathbf{x},t) - \vartheta_{2}(\mathbf{x},t)$, we write: 
\begin{align}
    \label{Eqn:GROM_Uniqueness_rad_step_2}
    \frac{\mathrm{d}}{\mathrm{d}t}\int_{\Omega}\frac{1}{2} \, \rho_{s} \, c_{s} \, w^{2}(\mathbf{x},t) 
    \, \mathrm{d} \Omega 
    + \int_{\Gamma^\mathrm{N}} w^2(\mathbf{x},t) \, \varepsilon \, \sigma \, 
    \big(\vartheta_1(\mathbf{x},t) + \vartheta_2(\mathbf{x},t) \big) \, \big(\vartheta^{2}_1(\mathbf{x},t) + \vartheta^{2}_2(\mathbf{x},t) \big) \,
    \mathrm{d} \Gamma
    \leq 0 
\end{align}
Since $\vartheta_{1}(\mathbf{x},t)$ and $\vartheta_{2}(\mathbf{x},t)$ are non-negative solutions, $\vartheta_1(\mathbf{x},t) + \vartheta_2(\mathbf{x},t) \geq 0$. Also, the term $\vartheta^{2}_1(\mathbf{x},t) + \vartheta^{2}_2(\mathbf{x},t)$ is patently non-negative. Further noting that $\varepsilon \geq 0$ and $\sigma > 0$, we conclude that the second integral is non-negative, yielding:
\begin{align}
    \label{Eqn:GROM_Uniqueness_rad_step_3}
    \frac{\mathrm{d}}{\mathrm{d}t}\int_{\Omega}\frac{1}{2} \, \rho_{s} \, c_{s} \, w^{2}(\mathbf{x},t) 
    \, \mathrm{d} \Omega 
    \leq 0 
\end{align}
The above inequality is the same as in Eq.~\eqref{Eqn:GROM_Uniqueness_norad_step_13}. The rest of the proof, again, remains the same as in Theorem \ref{Thm:GROM_Uniqueness_no_rad}. 
\end{proof}

The above uniqueness theorems provide confidence in the proposed mathematical model. 

%% file: Sections/S5_GROM_NR_Domain_BCs.tex
\section{NUMERICAL RESULTS: DOMAIN SIZE AND BOUNDARY CONDITIONS}
\label{Sec:S5_GROM_NR_Domain_BCs}

\begin{figure}[h]
    \centering
    \includegraphics[scale=0.55]{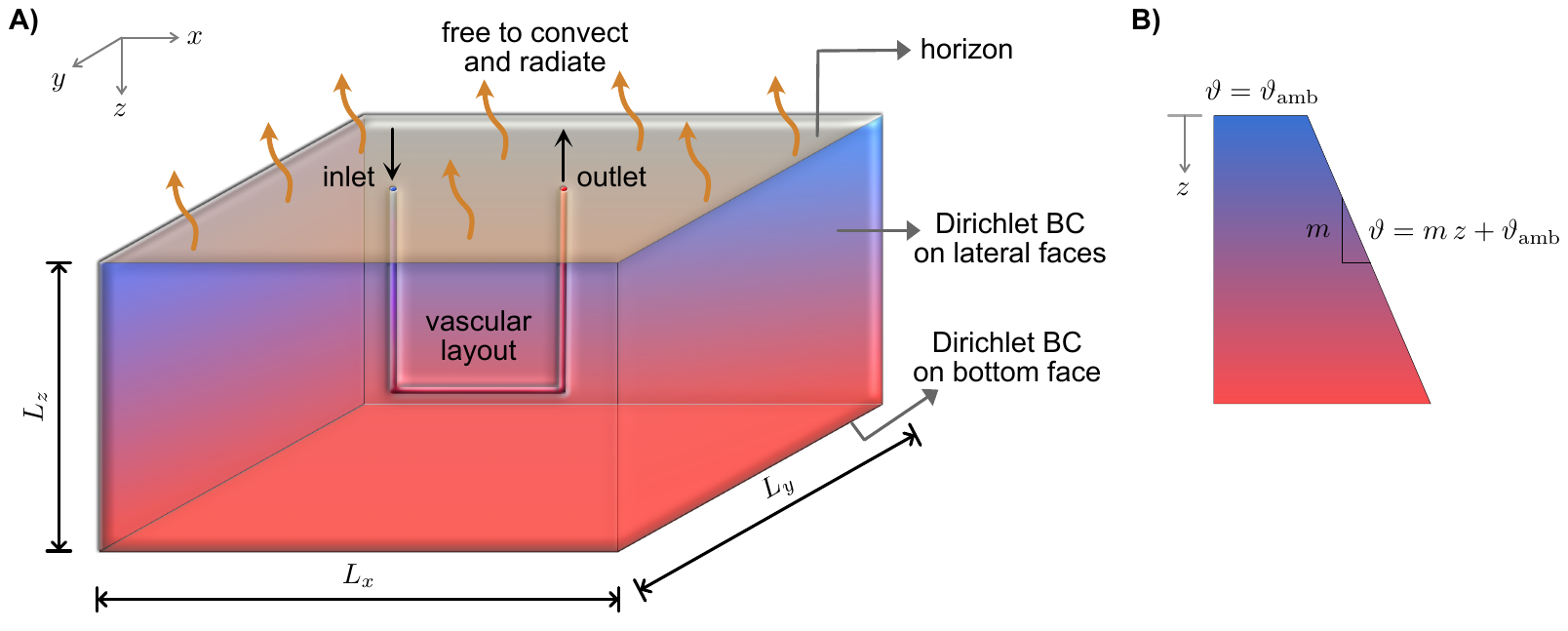}
    \caption{\textsf{Test problem}. \textbf{A)} The domain is a cuboid of size $L_x \times L_y \times L_z$,  featuring a vascular layout within its subsurface. The top face is free to convect and radiate, while the remaining faces (i.e., lateral and bottom faces) are subject to a Dirichlet boundary condition (BC). \textbf{B)} The prescribed temperature increases linearly with depth at a slope of $m$. \label{Fig3:GROM_Problem_setup}}
\end{figure}

We consider a representative initial boundary value problem in which the geothermal system is idealized as a cuboid with dimensions length $L_x$, width $L_y$, and depth $L_z$, as shown in \textbf{Fig.~\ref{Fig3:GROM_Problem_setup}}. For convenience, we use a Cartesian coordinate system with rectangular axes denoted by $(x,y,z)$. The top face---which is at the horizon---is free to convect and radiate. The bottom and lateral faces are subject to a prescribed temperature field; in mathematical parlance, these five subsurface faces are prescribed with Dirichlet boundary conditions. It is a good approximation to consider the temperature ($\vartheta$) increases linearly with the depth ($z$). Mathematically, we write this variation as follows: 
\begin{align}
    \label{Eqn:GROM_Variation_of_temp_with_depth}
    \vartheta = m z + \vartheta_{\mathrm{amb}}
\end{align}
where $m$ is the slope of the variation (geothermal gradient), and $\vartheta_{\mathrm{amb}}$ the ambient temperature. We use this variation to prescribe the initial and Dirichlet boundary conditions. Table \ref{Table:GROM_param_values} provides the parameters used in the numerical simulations.

\begin{table}[h]
    \centering 
    \caption{Parameters used in numerical simulations. \label{Table:GROM_param_values}} 
    \begin{tabular}{||c|c||}
        \hline
        \textbf{Parameter} & \textbf{Value} \\ [0.5ex] 
        \hline\hline
        Domain length ($L_x$) & 6000 m   
        \\ \hline 
        Domain width ($L_y$) & 6000 m   
        \\ \hline 
        Domain height ($L_z$) & 10000 m   
        \\ \hline 
        U-shape vascular depth ($d$) & 5000 m   
        \\ \hline 
        Comb-shape vascular depth ($d$) & 8000 m   
        \\ \hline 
        U-shape horizontal spacing ($s$) & 3000 m   
        \\ \hline 
        Comb-shape horizontal spacing ($s$) & 900 m  
        \\ \hline \hline
        Fluid & Water  
        \\ \hline 
        Fluid's density ($\rho_f$) & 1000 $\mathrm{kg/m^3}$ 
        \\ \hline
        Fluid's specific heat capacity ($c_f$) & 4183 $\mathrm{J/(kg \cdot K)}$ 
        \\ \hline 
        Mass flow rate ($\dot{m}$) & 5-60 $\mathrm{kg/sec}$ 
        \\ \hline \hline 
         Solid's density ($\rho_s$) & 2500 $\mathrm{kg/m^3}$ 
        \\ \hline 
        Solid's specific heat capacity ($c_s$) & 790 $\mathrm{J/(kg \cdot K)}$ 
        \\ \hline 
        Thermal conductivity ($k_{s}$) & 3.5 $\mathrm{W/(m \cdot K)}$
        \\ \hline \hline 
        Thermal gradient ($\mathrm{slope}$) & 30 $\mathrm{K/km}$   
        \\ \hline
        Heat transfer coefficient ($h_T$) & 0.5 $\mathrm{W/(m \cdot K)}$
        \\ \hline
        Ambient temperature ($\vartheta_{\mathrm{amb}}$) & 303.15 $\mathrm{K}$ \\
        \hline
        Inlet temperature ($\vartheta_{\mathrm{inlet}}$) & 303.15 $\mathrm{K}$ \\
        \hline \hline 
        Time-stepping scheme & backward difference formula (BDF) \\
        \hline
        Time step ($\Delta t$) & $1 \times 10^6$ secs \\ 
        \hline
        Total time ($\mathcal{T}$) & $2 \times 10^9$ secs \\
        \hline
    \end{tabular}
\end{table}

\subsection{Domain size} The selection of domain size depends on accurately capturing solution fields while managing computational costs. An obvious aspect of computational cost is the time required to obtain the solution. Additionally, modeling the Earth's surface, including buried channels, can make mesh creation both costly and time-intensive. Thus, determining an appropriate domain size that meets these requirements is crucial. This section aims to do just that.

To make the idea more tangible, we consider a domain of depth $d = 10$ km, comprising a U-shaped vascular layout with a constant channel spacing of $s = 3$ km. Both the length $L_x$ and width $L_y$ of the domain are determined by a constant multiplier of the spacing, denoted as $\alpha \, s$, resulting in $L_x = L_y = s + 2 \, \alpha \, s$. With $\alpha$ ranging from 0.5 to 2, this corresponds to a variation in length and width from 6 to 15 km. To examine the impact of domain size on the solution field, temperature data is plotted along a horizontal line at a depth of 0.7$d$ in the plane below the horizontal segment of the vascular layout, at the 63-year mark; see \textbf{Fig.~\ref{Fig4:GROM_domain_size}}. The length normalization takes the following mathematical form: 
\begin{align}
    \mathrm{Normalized \, length} 
    = \frac{x}{s}-(\alpha+0.5)
\end{align}
where $x$ is any point on the horizontal line.

The key criterion is that the domain size is deemed appropriate as long as the boundaries of the domain do not affect the solution field. When the value of $\alpha$ decreases from 2 to 0.5, the plots indicate that the solution fields converge and remain unaffected by the boundaries. The temperature remains consistent except for the two peaks occurring at the intersection of the horizontal line and the two legs of the U-shaped layout. A zoomed-in section of the plot further illustrates that, aside from the areas near the channels, the temperature across other parts remains uniform. Therefore, in response to our first scientific question (Q1), we conclude that $\alpha = 0.5$ (i.e., domain size $L_x = L_y =$ 6 km) can achieve a reliable numerical solution. Thus this value is used in all further studies.

\begin{figure}[h]
    \centering
    \includegraphics[scale=0.6]{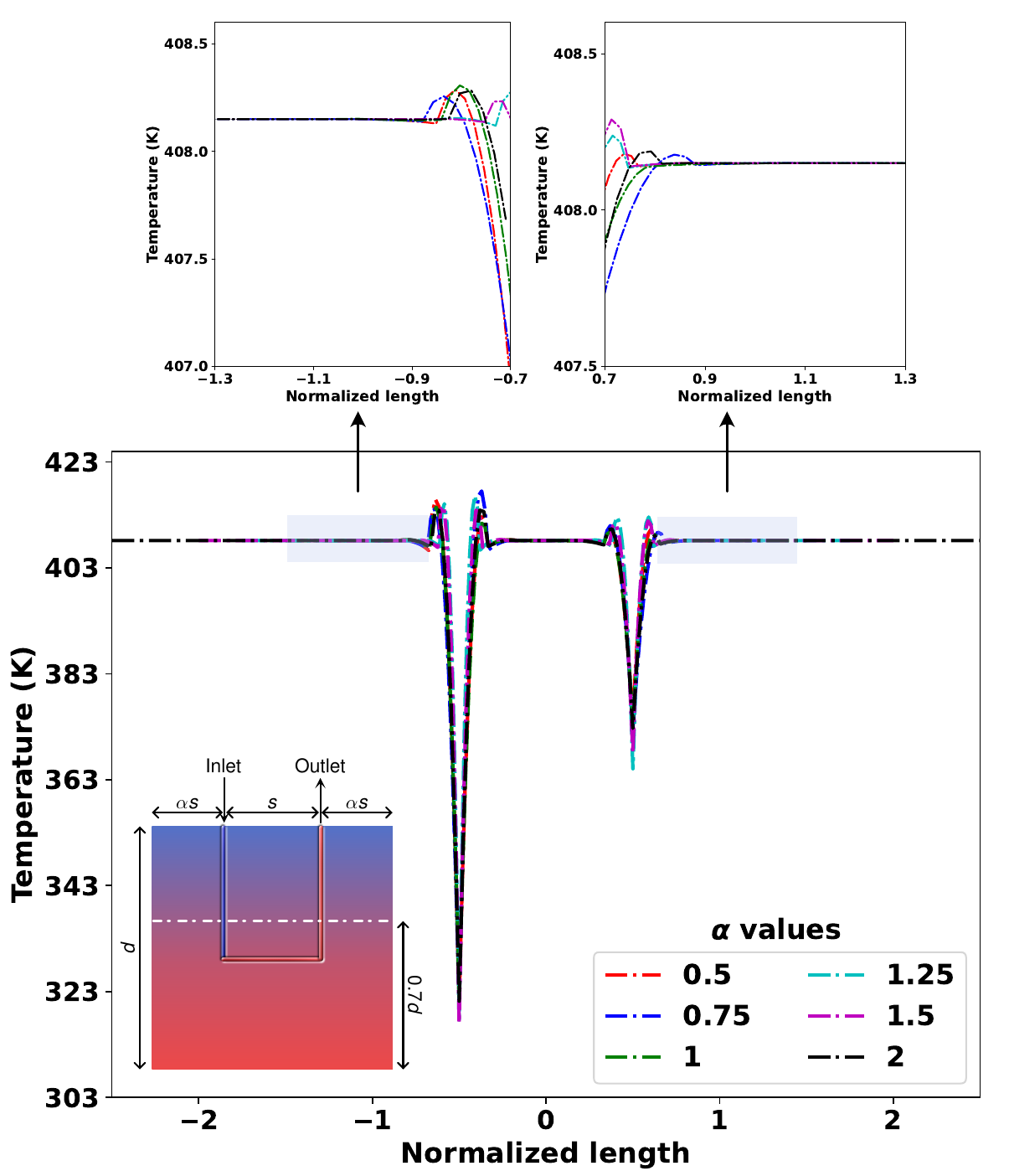}
    \caption{\textsf{Domain size:} This figure guides the selection of the domain size by observing its influence on the solution field. We consider a domain with depth $d$ and a horizontal dimension of $(1 + 2 \, \alpha)s$, where $s$ represents the spacing of the injection and production channels. For various values of $\alpha$ ranging from 0.5 to 2, we plotted the temperature at a depth of $0.7d$ after 63 years of geothermal system operation and for a mass flow rate of 35 kg/s. The solution profile remains mostly consistent within the chosen range of $\alpha$. \emph{Inference:} As a result, we recommend setting the domain size to be twice the spacing between the vertical legs in the vascular layout, corresponding to $\alpha = 0.5$.
    \label{Fig4:GROM_domain_size}}
\end{figure}

\subsection{Effect of boundary conditions}
We explore how the choice of boundary conditions impacts the solution field, specifically the outlet temperature in \textbf{Fig.~\ref{Fig5:GROM_BCs_check}}. We observe an invariance between the outlet temperature under two boundary conditions: Neumann and Dirichlet. This resemblance might be attributed to the far-field boundaries, meaning they are situated considerably from the channel through which the fluid flows. To simulate geothermal systems, our domain size should be large enough so that the boundary conditions do not significantly influence our solution field. The findings verify that the selected domain size for the geothermal system is robust, and the choice of boundary conditions does not influence the system's efficiency, answering our second science question (Q2).

\begin{figure}[h]
    \centering
    \includegraphics[scale=0.35]{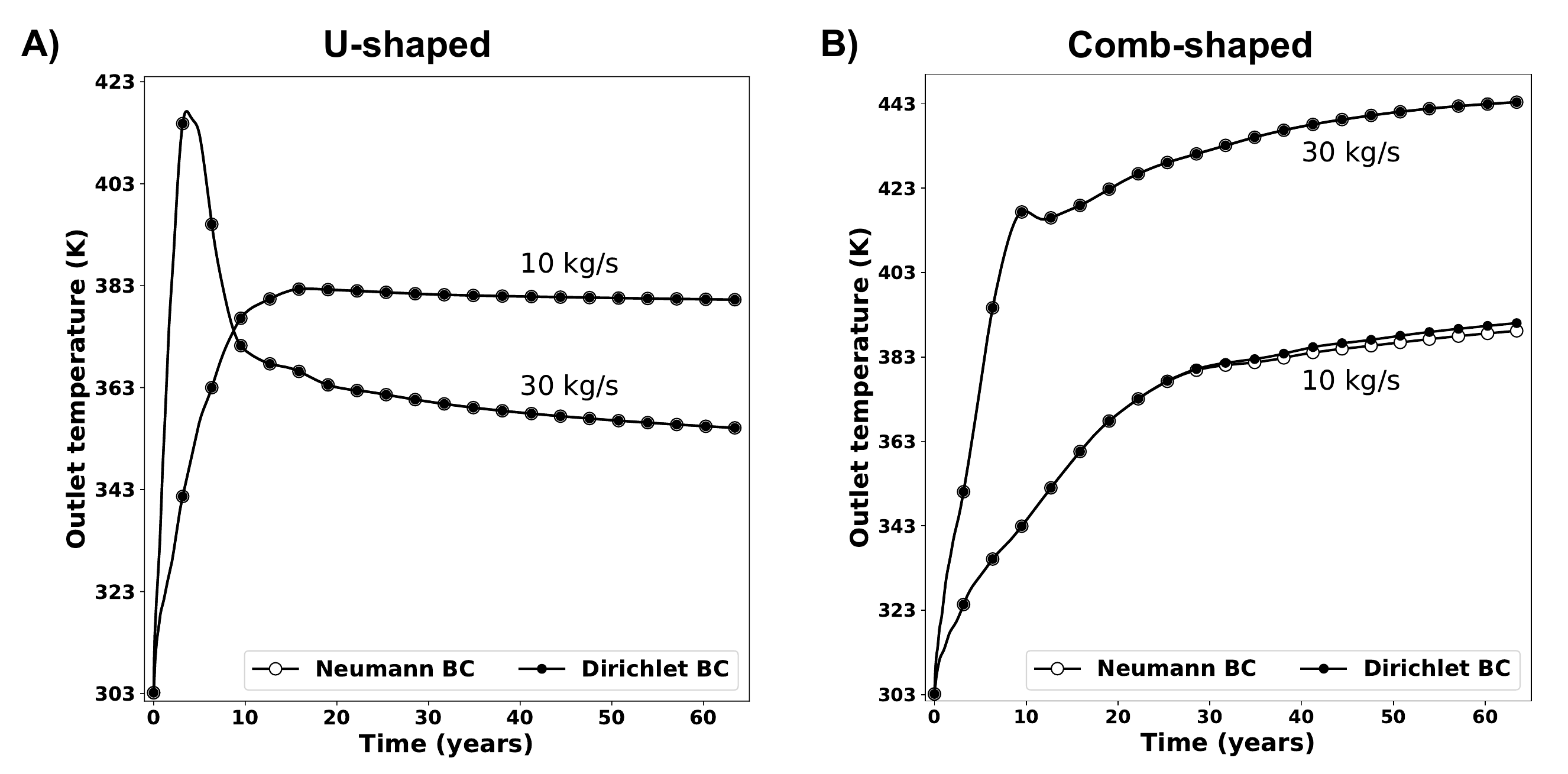}
    \caption{\textsf{Boundary conditions:} The figure reveals a close resemblance in the outlet temperatures under Neumann and Dirichlet boundary conditions. We considered two mass flow rates of 10 kg/s and 30 kg/s and plotted the temperatures over a 63-year duration. The results are reported for \textbf{A)} U-shaped and \textbf{B)} comb-shaped layouts. \emph{Inference:} We can conclude that the numerical solution from the model remains unaffected by the boundary conditions and behaves akin to a far-field condition.
    \label{Fig5:GROM_BCs_check}}
\end{figure}

\subsection{Temperature profile}
In \textbf{Fig.~\ref{Fig6:GROM_temperature_profile}}, we illustrate the temperature profile of the geothermal systems modeled in this study, aiming to enhance the readers' understanding. The profiles depicted correspond to the two vascular layouts outlined in Sect. \ref{Sec:S2_GROM_Setup}. In the \emph{exterior view}, the color scheme, starting with deep blue at the surface to deep red at the bottom, signifies an increase in temperature of 30 K/km as the depth into the subsurface increases. The two different vascular layouts embedded inside the domain through which coolant flows are shown under \emph{translucent view}. The water entering at the inlet with ambient temperature reaches its highest temperature at the maximum subsurface depth and experiences some heat loss as it returns to the outlet at the surface. The coloration near the layouts in the \emph{section view} distinctly represents the flow of fluid and its varying temperature as it moves along different points of the channel. Each sub-figure significantly reinforces the effectiveness of our model in depicting the geothermal system and the fluid flow within the channel.

\begin{figure}[h]
    \centering
    \includegraphics[scale=0.7]{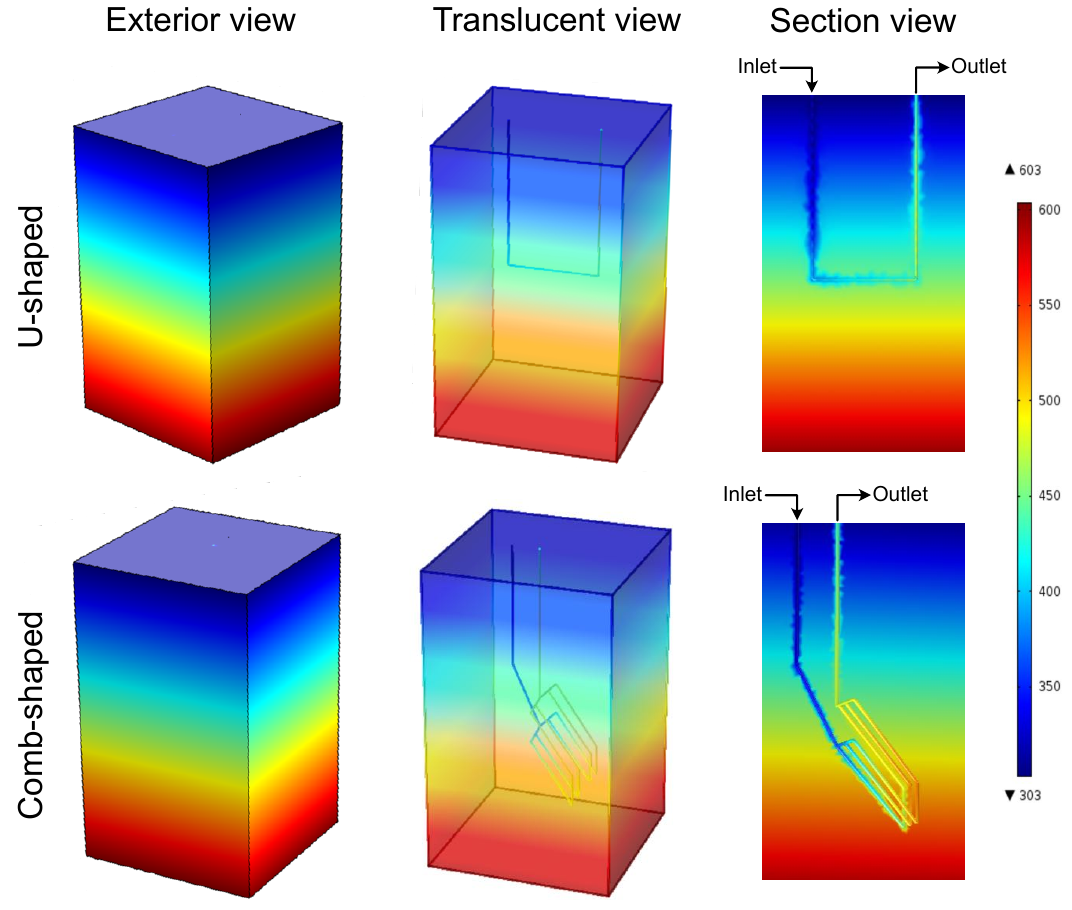}
    \caption{\textsf{3D temperature profile:} The figure illustrates the temperature profile within closed-loop geothermal systems for U-shaped (top panel) and comb-shaped (bottom) layouts. The temperature gradient is 30 K/km, with the temperature of the top surface equal to the ambient temperature (which is indicated in blue color in the graphics). A coolant (water) at ambient temperature enters the channel at the inlet. As it traverses along the vasculature, the water's temperature increases due to heat conduction. As a result, the outlet produces hot water, with the actual temperature varying depending on parameters and the vascular layout. As evident from the above figures, the proposed reduced-order model (ROM) successfully captures the detailed temperature distribution within the subsurface.
    \label{Fig6:GROM_temperature_profile}}
\end{figure}

%% file: Sections/S6_GROM_Efficiency.tex
\section{OUTLET TEMPERATURE AND EFFICIENCY}
\label{Sec:S6_GROM_Efficiency}

In this section, we study the efficiency of geothermal systems, focusing on production capacity---a key determinant of whether the system warrants further investigation and real-field application. While doing so, we introduce \emph{breakdown}, as a point of time at which the efficiency drops down for the system. We first examine the efficiency with respect to the outlet temperature and then analyze the power production, together answering the third scientific question (Q3).

\subsection{Outlet temperature}
As described in Sect.~\ref{Sec:S2_GROM_Setup}, both vascular layouts involve extracting a hot fluid at the outlet. The coefficient of performance of the system, given in Eq. \eqref{Eqn:GROM_coeff_of_performance}, increases with an increase in the outlet temperature $\vartheta_{\mathrm{outlet}}$, as $\vartheta_{\mathrm{inlet}}$ is constant (cf. Eq. \eqref{Eqn:GROM_GE_Inlet}). Consequently, a plot of $\vartheta_{\mathrm{outlet}}$ over time can effectively represent the system's coefficient of performance. \textbf{Figure~\ref{Fig7:GROM_Outlet_temperature}} illustrates the outlet temperature measurements for both vascular layouts over a 63-year period and under various mass flow rates. One can make the following inferences from the plot. In the case of the \emph{U-shaped} layout in \textbf{Fig.~\ref{Fig7:GROM_Outlet_temperature}}A), the peak outlet temperature increases as the mass flow rate increases reaching as high as 420 K for 30 kg/s. However, lower mass flow rate values, such as 5 kg/s and 10 kg/s, do not lead to the system's breakdown, but the higher mass flow rates cause the breakdown sooner: 15 kg/s at 10 years, 20 kg/s at 7 years, 25 kg/s at 5 years, and 30 kg/s at 4 years.

In the case of the \emph{comb-shaped} layout in \textbf{Fig.~\ref{Fig7:GROM_Outlet_temperature}}B), as the mass flow rate increases, the peak outlet temperature also rises, reaching a maximum of 440 K for 30 kg/s. Also, none of the mass flow rates lead to a breakdown, offering a significant advantage over the U-shaped layout and the potential for substantially higher average power production.

\begin{figure}[h]
    \centering
    \includegraphics[scale=0.375]{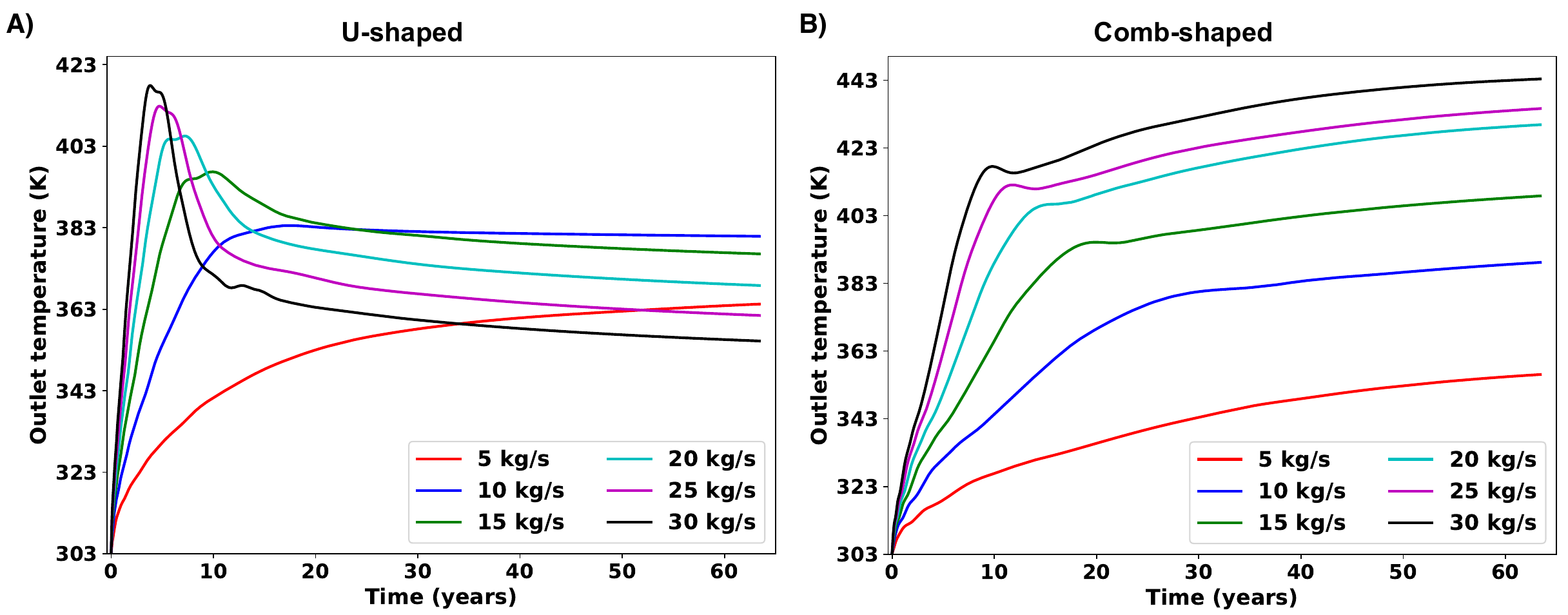}
    \caption{\textsf{Outlet temperature.} The figure collates the results from a parametric study on the temporal evolution of the outlet temperature for different mass flow rates. \textbf{A)} In a U-shaped layout, an increase in mass flow rate leads to a rise in outlet temperature. However, a higher mass flow rate results in an earlier breakdown. \textbf{B)} On the other hand, in a comb-shaped layout, the outlet temperature rises with increasing mass flow rates. Notably, none of the curves, for various mass flow rates, reaches a breakdown point within the 63-year time frame. \emph{Inference:} Given the higher outlet temperature and the absence of breakdown, the comb-shaped outperforms the U-shaped for higher mass flow rates.
    \label{Fig7:GROM_Outlet_temperature}}
\end{figure}

\subsection{Power production} To gain a deeper insight into the efficiency of the geothermal system, we further examine the power production over time and the average power produced throughout the system's life, which is 63 years in our study. The power produced at any time is calculated utilizing the following formula:
\begin{align}
    \label{GROM_Eqn_Power}
    \mathrm{Power}(t) = \dot{m} \, c_f \, (\vartheta_\mathrm{outlet}(t) - \vartheta_\mathrm{amb}) \,\,\,\, \forall t\in [0,\mathcal{T}]
\end{align}
The power production over time for both vascular layouts is depicted in \textbf{Fig.~\ref{Fig8:GROM_Power_production}}A and B, respectively. These plots exhibit a monotonic and smoother trend than the previous outlet temperature plot. In both layouts, the power output increases with a higher mass flow rate. However, the U-shaped layout reaches the breakdown phase earlier ($\approx$ 10 years period). However, the comb-shaped layout does not encounter breakdown during the whole 63-year time period.

Finally, the average power production of the geothermal system is calculated using the following expression:
\begin{align}
    \label{GROM_eqn_average_power}
    \mathrm{Average \, power} = \frac{1}{\mathcal{T}}\int_{0}^{\mathcal{T}} \, \mathrm{Power}(t)\, dt
\end{align}
\textbf{Figure~\ref{Fig8:GROM_Power_production}}C illustrates the average power production over 63 years lifetime of both the layouts, showcasing a wide range of mass flow rate variations, extending up to 60 kg/s. Below are some key observations from the bar plot:
\begin{enumerate}[(1)]
    \item For lower mass flow rates of 5 kg/s and 10 kg/s, the average power of the U-shaped layout exceeds that of the comb-shaped layout. Nonetheless, the amount of average power of both layouts is not significant for practical field applications.
    \item Notably, beyond 25 kg/s, the relative increment in average power production for the U-shaped layout is negligible, while that of the comb-shaped layout is increasing exponentially.
    \item Given the intricate branching and higher depth of the comb-shaped layout, the fluid path delves further into the hot subsurface, resulting in significantly higher average power production than the U-shaped layout.
    \item It is noteworthy that the U-shaped layout shows admirable performance up to 20--25 kg/s, given its simple shape and ease of field application. However, the comb-shaped layout excels at higher mass flow rates, outperforming the U-shaped layout.
\end{enumerate}

\begin{figure}[h]
    \centering
    \includegraphics[scale=0.4]{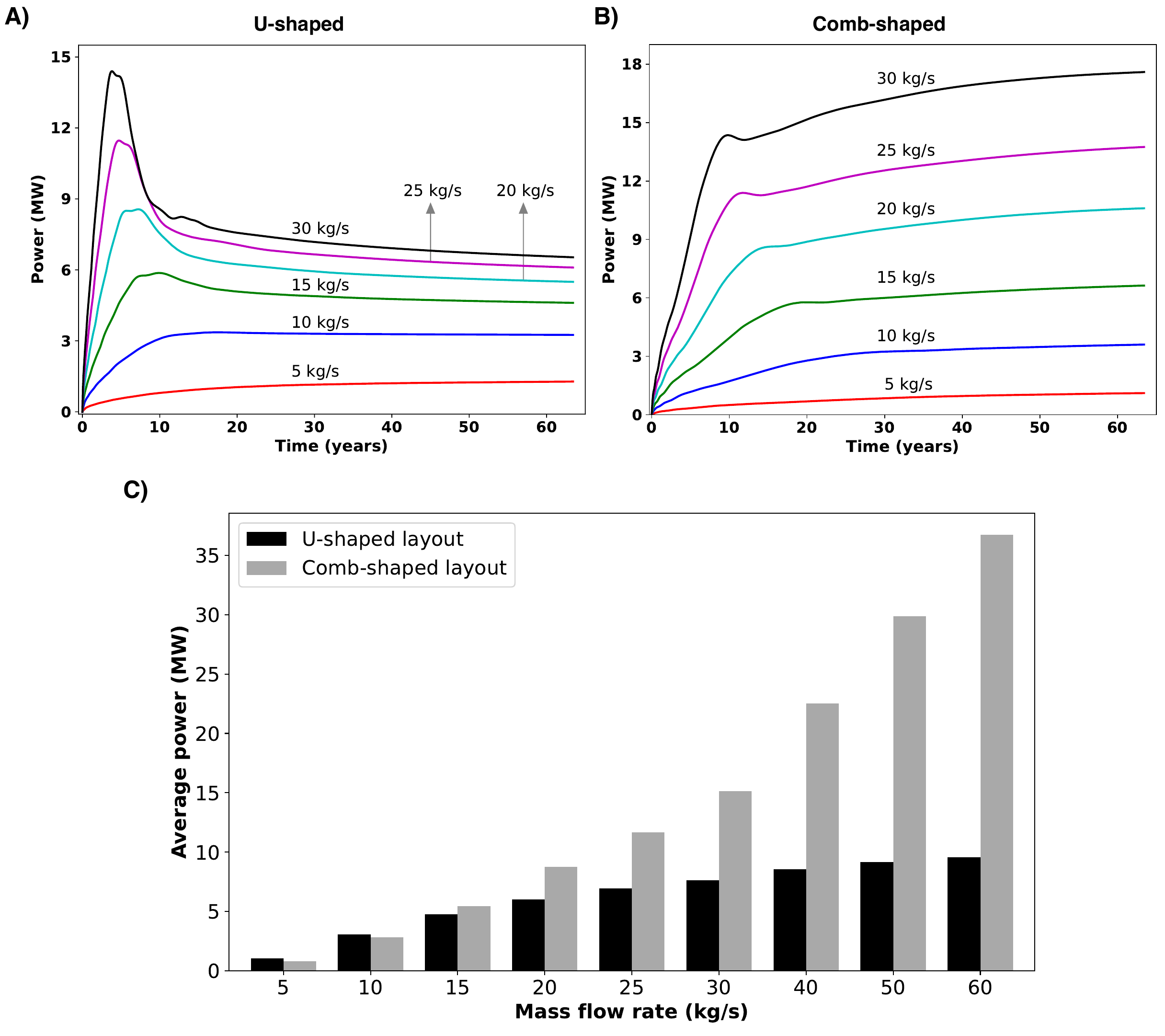}
    \caption{\textsf{Power production.} The figure illustrates the evolution of power production over a span of 63 years for various mass flow rates: \textbf{A)} U-shaped and \textbf{B)} comb-shaped layouts. For both vascular layouts, power production increases with mass flow rate. However, the U-shaped layout experiences a significant drop in power production around the 10-year mark due to a breakdown, reducing its output. \textbf{C)} The average power bar chart offers a clear representation of the geothermal system's efficiency. The U-shaped layout demonstrates better performance at lower mass flow rates. However, as the mass flow rate increases, the comb-shaped layout becomes more efficient, generating approximately four times the average power of the U-shaped layout at 60 kg/s. \emph{Inference:} For the U-shaped layout, a lower mass flow rate (no greater than 20--25 kg/s) is preferable, while a higher mass flow rate is recommended for the comb-shaped layout.
    \label{Fig8:GROM_Power_production}}
\end{figure}

%% file: Sections/S7_GROM_Surface_MST.tex
\section{MEAN TEMPERATURE ON THE TOP SURFACE}
\label{Sec:S7_GROM_Surface_MST}

\textbf{Figure \ref{Fig9:GROM_MST_top_surface}}  illustrates the development of surface temperature over time, starting from the inception of a U-shaped geothermal system. Various sizes for the region near the outlet, ranging from 100 m to 450 m, were considered, as shown in Fig.~\ref{Fig9:GROM_MST_top_surface}A. Using numerical results with a high flow rate of 60 kg/s, we calculated the mean surface temperature within each region individually. Figure~\ref{Fig9:GROM_MST_top_surface}B demonstrates that the increase in mean surface temperature is negligible, amounting to less than one degree for regions larger than 200 m, further proving that geothermal energy entails minimal adverse effects. Lastly, this study addressed the fourth scientific question (Q4) outlined in the introduction.

\begin{figure}[h]
    \centering
    \includegraphics[scale=0.42]{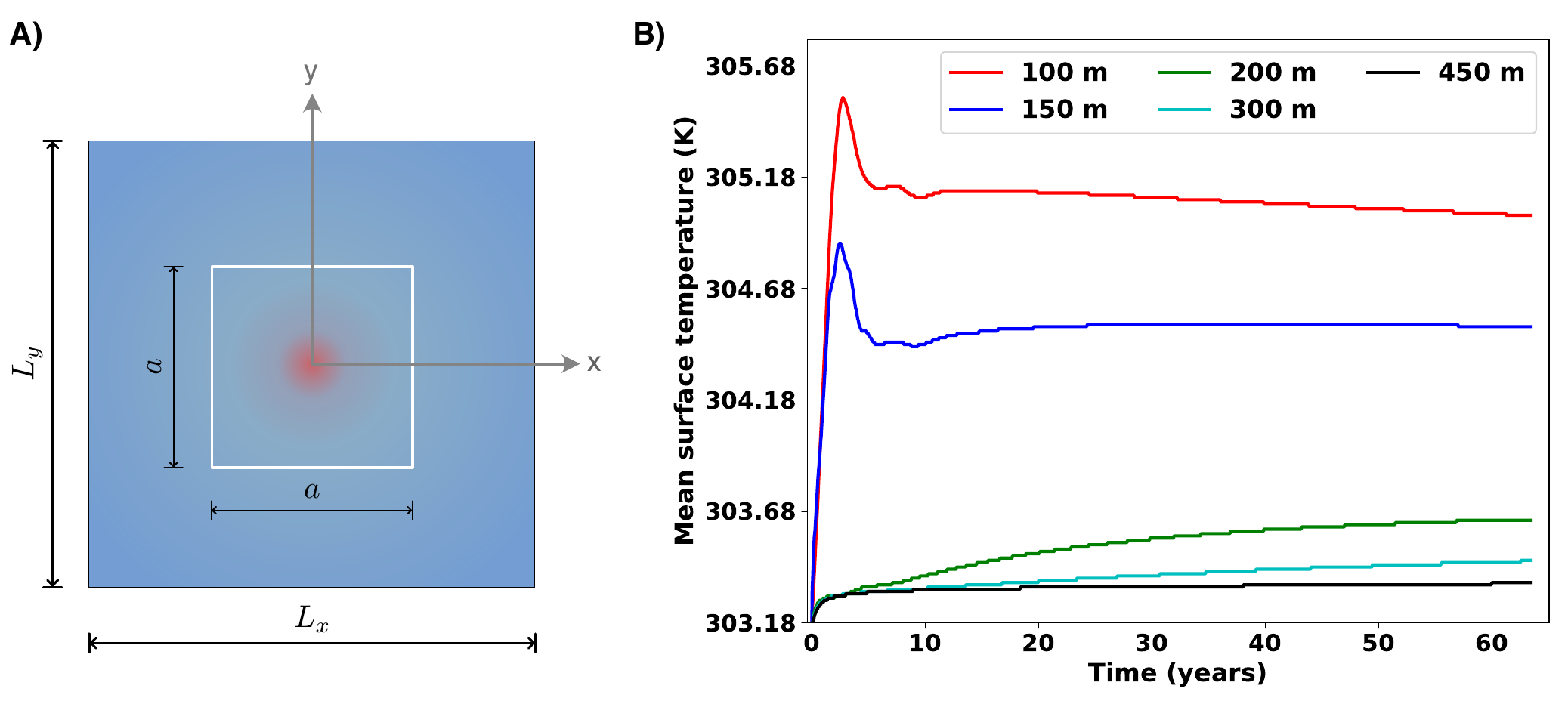} 
    \caption{\textsf{Effect on mean surface temperature.} The figure answers the question: How does a working geothermal system affect the mean surface temperature of regions near the extraction well (i.e., the outlet)?  
    \textbf{A)} We considered various regions near the outlet, whose dimensions are given by the edge length $(a)$. We selected five sizes varying from 100 m to 450 m. \textbf{B)} The mean surface temperature is calculated for each region and is plotted over the time since the inception of the production. The mass flow rate is taken as 60 kg/s. \emph{Inference:} For the square region with $a = 100$ m, a temperature increase of about two degrees K is noted. For larger regions (i.e., $a \geq 200$ m), the rise in mean surface temperature is marginal (less than 1 K). \label{Fig9:GROM_MST_top_surface}}
\end{figure}

%% file: Sections/S8_GROM_Discussion_Closure.tex
\section{DISCUSSION AND CLOSURE}
\label{Sec:S8_GROM_Discussion_Closure}
We have presented a novel modeling framework to study the transient response of closed-loop geothermal systems---also referred to as advanced geothermal systems (AGS). The framework comprises a reduced-order mathematical model, which treats the embedded vasculature as a curve, and an associated numerical formulation based on the finite element method. The significance of our research lies in providing a reliable tool for accurately forecasting the subsurface temperature field and the outlet temperature, which can be utilized to calculate long-term power production rates and efficiency. 

In the realm of geothermal energy, various modeling approaches have been explored for different systems, including EGS and AGS. However, a pivotal question remains: \emph{Whether the model is predictive?} Posed differently: Are the model's results reliable? Our ROM model for geothermal applications has been built based on the mathematical model developed for thermal regulation in microvascular composites. For that application, which assumed the domain to be slender, \citet{nakshatrala2023ROM} has provided the mathematical underpinning, presenting several qualitative properties that establish the model's reliability and trustworthiness. Also, in active cooling of microvascular composites, the mathematical model has been validated experimentally \citep{2023_Nakshatrala_PNAS_Nexus}. This combination of theoretical robustness and practical verification underscores the reliability of our ROM in the study of geothermal systems. Therefore, our proposed model can be considered predictive. 

The mathematical model and numerical results presented in the paper provide a deeper understanding of AGS and aid in perfecting such energy systems, contributing to the development of sustainable green energy solutions on a large scale. Some of the \emph{salient features} of the proposed modeling framework are: 
\begin{enumerate}[(S1)]
\item The ROM framework provides a quick and accurate prediction of the temperature field, including the outlet temperature, which is proportional to thermal efficiency. 
\item The model accurately incorporates the jump conditions across the vasculature. 
\item The model enjoys several attractive mathematical properties, including the uniqueness of solutions. 
\item The framework enables users to determine the system’s performance and optimal capacity.
\end{enumerate}

The \emph{main findings} of this work are: 
\begin{enumerate}[(C1)]
    \item For a given spacing between the inlet and outlet branches, the lateral dimensions of the computational domain can be taken as twice this distance.
    \item The higher mass flow rates result in higher outlet temperature and, subsequently, higher power production.
    \item The time it takes for the system to reach a peak depends on the mass flow rate.
    \item The efficiency of geothermal systems in terms of average power is calculated. The U-shaped layout shows better performance at lower mass flow rates, while the comb-shaped layout excels significantly at higher mass flow rates.
    \item The study also confirms that geothermal systems do not adversely impact neighboring lands, as demonstrated through an analysis of the mean surface temperature (MST) of the top surface.
\end{enumerate}

In a nutshell, our findings indicate that lower mass flow rates are more effective with a U-shaped layout, whereas a comb-shaped layout is optimal with higher mass flow rates. However, considering the depth and layout design, the comb-shaped layout demands substantially higher resources than the U-shaped layout. By carefully evaluating these factors, stake holders can make informed decisions to optimize the design and operation of geothermal systems, ultimately achieving maximum effectiveness and efficiency.

Finally, we envision two fronts for potential future work. The first one is to extend the model by incorporating two-phase characteristics of the flowing fluids. This will be particularly important when supercritical carbon dioxide is used. The second one is to develop a machine-learning-based model that can offer fast forecasting, which is ideal for real-time monitoring of a geothermal system. 